\documentclass[10pt,journal]{IEEEtran}
\usepackage{subfigure}
\usepackage{tabulary}
\usepackage{footnote}
\usepackage{optidef}
\usepackage[export]{adjustbox}
\usepackage{colonequals}
\usepackage{comment}
\usepackage{cases}
\usepackage{multicol}
\usepackage[justification=centering]{caption}
\graphicspath{{./matlab/figures/}{./}{figures/}}

\usepackage{lipsum}  
\usepackage{amsfonts}
\usepackage{mathtools}
\usepackage{bm}
\usepackage{amsmath,amssymb}
\usepackage{algorithm}
\usepackage{algpseudocode}
\usepackage{amsthm}
\usepackage{booktabs}
\usepackage{tabularx}
\usepackage{multirow}
\usepackage{graphicx}
\usepackage{xcolor}
\usepackage{cite}
\usepackage{hyperref}
\usepackage{tikz}
\usetikzlibrary{arrows.meta,positioning,fit,calc,shapes.geometric,shapes.misc,decorations.pathmorphing,backgrounds,matrix}
\newcommand{\safeincludegraphics}[2][]{\IfFileExists{#2}{\includegraphics[#1]{#2}}{\fbox{\parbox[c][0.28\textheight][c]{0.94\textwidth}{\centering Figure file not found: \texttt{\detokenize{#2}}}}}}
\hypersetup{hidelinks}

\DeclareMathOperator*{\argmax}{argmax}

\DeclareMathOperator{\Tr}{Tr}
\DeclareMathOperator{\rank}{rank}
\DeclareMathOperator{\diag}{diag}

\newtheorem{lemma}{Lemma}
\newtheorem{proposition}{Proposition}
\newtheorem{theorem}{Theorem}
\newtheorem{corollary}{Corollary}

\newcommand{\C}{\mathbb{C}}
\newcommand{\R}{\mathbb{R}}
\newcommand{\E}{\mathbb{E}}
\newcommand{\CN}{\mathcal{CN}}
\newcommand{\cM}{\mathcal{M}}
\newcommand{\cE}{\mathcal{E}}
\newcommand{\cL}{\mathcal{L}}
\newcommand{\cS}{\mathcal{S}}

\newcommand{\cB}{\mathcal{B}}

\newcommand{\bH}{\mathbf{H}}
\newcommand{\bF}{\mathbf{F}}
\newcommand{\bA}{\mathbf{A}}

\newcommand{\bR}{\mathbf{R}}
\newcommand{\bQ}{\mathbf{Q}}
\newcommand{\bJ}{\mathbf{J}}

\newcommand{\bC}{\mathbf{C}}
\newcommand{\bI}{\mathbf{I}}
\newcommand{\bzero}{\mathbf{0}}
\newcommand{\bx}{\mathbf{x}}
\newcommand{\bz}{\mathbf{z}}
\newcommand{\bv}{\mathbf{v}}
\newcommand{\bu}{\mathbf{u}}
\newcommand{\br}{\mathbf{r}}
\newcommand{\bs}{\mathbf{s}}
\newcommand{\be}{\mathbf{e}}
\newcommand{\by}{\mathbf{y}}

\newcommand{\bw}{\mathbf{w}}
\newcommand{\bq}{\mathbf{q}}
\newcommand{\bo}{\mathbf{o}}

\newcommand{\bdelta}{\boldsymbol{\delta}}
\newcommand{\bPhi}{\boldsymbol{\Phi}}
\newcommand{\bOmega}{\boldsymbol{\Omega}}

\title{Logit-Aware MIMO AirComp for Distributed Mixture-of-Experts LLM Inference over Wireless Edge Networks}

\author{\text{Lyutianyang Zhang},~\IEEEmembership{Member,~IEEE},
\text{Yunjian Jia}, \text{Liu Cao},~\IEEEmembership{Member,~IEEE}, \text{Dengke Wang}, \text{Jinke Ren},~\IEEEmembership{Member,~IEEE}, and \text{Shuguang Cui, }\IEEEmembership{Fellow, IEEE}
\thanks{Lyutianyang Zhang, Yunjian Jia, and Dengke Wang are with the School of Microelectronics and Communication Engineering, Chongqing University, Chongqing, China (email: \{zhanglyutianyang, yunjian, dengke.wang\}@cqu.edu.cn). Liu Cao is with City University of Hong Kong (Dongguan), Dongguan, Guangdong, China (email: liu.cao@cityu-dg.edu.cn). Jinke Ren and Shuguang Cui are with the School of Science and Engineering, the Shenzhen Future Network of Intelligence Institute, and the Guangdong Provincial Key Laboratory of Future Networks of Intelligence, the Chinese University of Hong Kong (Shenzhen), Shenzhen, Guangdong 518172, China (e-mail: \{jinkeren, shuguangcui\}@cuhk.edu.cn). (Corresponding author: Yunjian Jia).}}

\begin{document}
\maketitle

\begin{abstract}
Distributed mixture-of-experts (MoE) inference is a promising architecture for deploying large language models (LLMs) at wireless edge networks because sparse experts can be placed across coordinated base stations (BSs), while the anchor node and user equipment (UE) can offload LLM inference tasks to BSs. The communication bottleneck is the MoE aggregation where the anchor BS must recover a weighted sum of selected expert outputs before each decoding step. Over-the-air computation (AirComp) is well matched to this operation because the wireless multiple-access channel naturally superposes simultaneous transmissions. However, conventional AirComp minimizes communication distortion, whereas MoE aggregation errors are not equally important to the LLM. The same hidden-state error can cause different next-token logit changes depending on the layer, token, and etc. In this paper, we propose a logit-aware multiple-input and multiple-output (MIMO) AirComp framework that estimates local logit sensitivity as block-level weights and jointly optimizes receive combiners and BS precoders under per-BS power constraints. We propose an alternating algorithm that protects aggregation from perturbation to the decoding decision. In the simulation, we design a wireless transceiver simulator for distributed MoE LLM inference. OpenCompass is utilized to evaluate the performance of Qwen3-30B-A3B-Instruct-2507-FP8 under GSM8K and ARC-Challenge datasets. At 30~dB aggregation SNR in the controlled replay, perturbed Qwen3 inference retains $99.3\%$ of the clean GSM8K accuracy and $94.8\%$ of the clean ARC-Challenge accuracy. In a direct ARC-Challenge closed-loop audit over 295 examples, the proposed sensitivity weight (SW)-AirComp algorithm achieves $43.73\%$ accuracy at 30~dB and $26.78\%$ at 20~dB, outperforming unweighted, matched-filter, and zero-forcing AirComp at both signal-to-noise ratios (SNRs). In the wireless simulator, SW-AirComp reduces decision-relevant aggregation distortion relative to unweighted, matched-filter, and zero-forcing baselines; in the 20~dB fixed-weight convergence diagnostic, its final weighted-sum mean-squared error (WSMSE) is $53.5\%$ lower than unweighted AirComp under the same channels, samples, and sensitivity weights. Logit-RMSE and logit-gap diagnostics further show that the gain comes from reducing output-logit perturbation and lowering the risk of top-token changes, not merely from lowering average hidden-state error.
\end{abstract}


\section{Introduction}

Transformer-based large language models (LLMs) have become the dominant backbone for language understanding, reasoning, and content generation~\cite{vaswani2017attention}. Moving these models from remote clouds to wireless edge networks is attractive because it can reduce response latency, keep user data closer to the access network, and exploit the storage and computation resources already deployed across base stations (BSs). The difficulty is that modern LLM inference is both memory intensive and delay sensitive. During autoregressive decoding, the model must finish the computation for one token before the next token can begin, while the key-value (KV) cache and intermediate activations continue to grow with context length. Therefore, a single edge node becomes a fragile serving point for LLMs.

Sparse mixed-of-experts (MoE) models solve this deployment problem in a useful way. The sparsely gated MoE layer activates only a small subset of experts for each token~\cite{shazeer2017moe}. GShard showed that conditional computation and automatic sharding can scale LLMs across multiple devices~\cite{lepikhin2020gshard}. Switch Transformer further demonstrated that simplified sparse routing can support trillion-parameter-scale models with manageable per-token computation~\cite{fedus2022switch}. These properties make MoE inference a natural fit for coordinated wireless edge clusters: experts can be stored at different BSs, while an anchor BS keeps the router, output head, KV cache, and decoding state. For each MoE layer, the selected BSs compute their weighted expert outputs, and the anchor BS needs only the sum of expert outputs before the next layer can proceed.

This observation exposes the main communication primitive in distributed MoE inference: the anchor BS does not need to decode independent messages; it needs a weighted sum of selected expert activations. Over-the-air computation (AirComp) is well matched to this primitive because its simultaneous transmissions over a wireless multiple-access channel are naturally superposed at the receiver~\cite{weijie}\cite{zhaohui}. In principle, the wireless channel can transform many expert-output transmissions into one analog aggregation step, reducing the token-level communication burden of distributed MoE inference. However, directly applying conventional AirComp to MoE inference misses a central feature of LLM computation. The aggregated hidden vector is not the final answer; it is an intermediate state that will pass through later attention, normalization, experts, and output-head layers. As a result, two hidden-state errors with the same Euclidean energy can have very different effects on next-token logits. An error in a sensitive direction can change the top-token decision, whereas a larger error in an insensitive direction may be absorbed by the remaining network.

The problem studied in this paper is therefore not ordinary AirComp MSE minimization. 
We consider logit-aware MIMO AirComp for a pretrained MoE model, where the receiver should recover sum of expert outputs in a way that preserves the next-token logits \footnote{A logit is the unnormalized score assigned by the LLM output head to each vocabulary token before normalization and token selection.}. The wireless design must handle channel fading, multi-antenna precoding, resource-block coupling, and per-BS power budgets, while the model side supplies a nonuniform sensitivity profile over aggregation blocks. This leads to the following design question: \emph{how should communication and computation resources be allocated when the important quantity is not the average hidden-state error, but the downstream logit distortion caused by that error?}

The key intuition is an error-budget mismatch between conventional AirComp and LLM inference. A conventional AirComp receiver treats the hidden state aggregation mainly as a communication signal and therefore attempts to reduce its average reconstruction error. A decoder, however, only observes this error after it has been transformed by the remaining pretrained layers and the output head. Some hidden states are almost neutral for the next-token distribution, while others move probability mass across competing vocabulary tokens. The relevant decision quantity is the gap between the largest and second-largest clean logits: if the aggregation-induced logit perturbation exceeds this gap, the selected token can change. A logit-aware AirComp design should therefore protect aggregation components with large downstream logit sensitivity rather than allocating equal power and spatial degrees of freedom on all hidden coordinates.

\subsection{Literature Review}
Wireless edge intelligence has established that learning performance and radio resource allocation are tightly coupled. A survey of distributed learning over wireless networks identified communication efficiency as a key bottleneck for wireless edge intelligence~\cite{chen2021distributed}. Wireless federated learning (FL) over resource-constrained networks has been compared under digital and analog transmission schemes~\cite{yao2024wirelessFL}. The network performance improved by the edge intelligence can be significantly affected by the wireless network features such as channel model \cite{zhang1} and network topology \cite{zhang2}. These works mainly address generic edge learning, where the transmitted objects are gradients, model updates, or compressed features. Distributed LLM inference is different: the model is fixed, the loop is autoregressive, and the transmitted object is an internal hidden activation whose importance is judged by the next generated token.

AirComp provides the physical-layer foundation for low-latency aggregation. Massive-IoT AirComp clarified over-the-air aggregation as a wireless computing primitive for dense networks~\cite{zhu2021massiveiot}. Power control for AirComp over fading channels was optimized to reduce aggregation distortion under channel uncertainty~\cite{cao2020powercontrol}. Optimization and scaling laws further characterized how AirComp behaves as the system dimension and device population grow~\cite{liu2020scaling}. These studies minimize aggregation distortion in the communication domain, which is necessary but not sufficient for aggregating LLM hidden state.

AirComp has also been widely used for FL. Over-the-air FL aggregation was proposed to reduce communication latency during distributed training~\cite{yang2020aircomp}. Broadband analog aggregation was designed for low-latency FL~\cite{zhu2020broadband}. FL over wireless fading channels has been analyzed to connect channel impairment with training behavior~\cite{amiri2020fading}. One-bit over-the-air aggregation was developed with convergence analysis~\cite{zhu2021onebit}. Compressive sensing has been used to reduce the dimension of massive-MIMO federated updates~\cite{jeon2021compressive}. Misaligned AirComp was addressed for FL under synchronization mismatch~\cite{shao2022misaligned}. Massive digital AirComp was studied for communication-efficient FL~\cite{qiao2024massive}. These works primarily support training or model-update aggregation; none of them models how an aggregation error propagates through the remaining layers of a pretrained LLM.

MIMO and programmable propagation environments improve AirComp reliability for high-dimensional signals, while multicarrier resource blocks provide a practical implementation method for broadband channels. Hybrid beamforming has been designed for massive-MIMO AirComp~\cite{zhai2021hybrid}. Transceiver beamforming has been studied for AirComp in massive-MIMO systems~\cite{jing2023transceiver}. Intelligent reflecting surfaces have been used to reshape the AirComp channel~\cite{jiang2021irs}. Reconfigurable intelligent surfaces have also been considered for improving AirComp reliability~\cite{fang2021ris}. Bayesian AirComp introduced statistical prior information into the aggregation design~\cite{shao2023bayesian}. Task-oriented AirComp was proposed for multi-device edge intelligence~\cite{wen2024taskaircomp}. These works are close in spirit because they go beyond scalar aggregation, but their task models do not describe distributed MoE expert summation or next-token logit sensitivity.

Wireless LLM inference is now receiving increasing attention \cite{zhang2025edgellm,lyu2026larger}. For example, digital and analog wireless FL aggregation under resource constraints are developed for LLMs~\cite{yao2024wirelessFL}. Physical-layer digital-twin models are developed from an ISAC perspective~\cite{cui2023digitalTwinISAC}. LLM-MoE transmission is studied for satellite networks~\cite{zhang2024satellite}. Recent inference and split-learning studies further address MIMO AirComp and wireless split learning~\cite{yang2025mimoInference,kim2024splitmac}. These studies motivate LLM over wireless edge networks, but they do not formulate physical (PHY)-layer aggregation of distributed MoE expert activations as a logit-aware AirComp problem.

\subsection{Main Contributions}
The main contributions are summarized as follows.
\begin{itemize}
    \item We directly connect wireless communication design with LLM inference by deriving logit sensitivity of MoE aggregation errors and embedding them into a MIMO AirComp design under per-BS power constraints. The resulting logit-weighted objective targets next-token preservation rather than average hidden-state reconstruction error, while the frequency-domain resource-block model provides the implementation framework.

    \item We develop a distributed MoE inference model for coordinated BS clusters, where selected expert outputs are aggregated through MIMO AirComp. A Jacobian-based perturbation analysis explains why the aggregation errors with the same hidden-state can have different effects on the next-token logits and inference decisions.

    \item We formulate and analyze the logit-sensitivity-weighted transceiver problem, deriving the AirComp MSE decomposition, logit-distortion bounds, token-preservation conditions, convex block subproblems, Karush-Kuhn-Tucker (KKT)-based closed-form updates, and monotonic convergence of the fixed-weight inner loop.

    \item We evaluate the framework through Qwen3 perturbation replay with OpenCompass dataset, a direct closed-loop ARC-Challenge audit, and a wireless simulator over 16 aggregation blocks. The proposed method is compared with unweighted, matched-filter, and zero-forcing AirComp using task accuracy, WSMSE, logit-RMSE, and logit-gap diagnostics.
\end{itemize}

\begin{figure*}[t]
\safeincludegraphics[height=0.61\textwidth,width=0.99\textwidth]{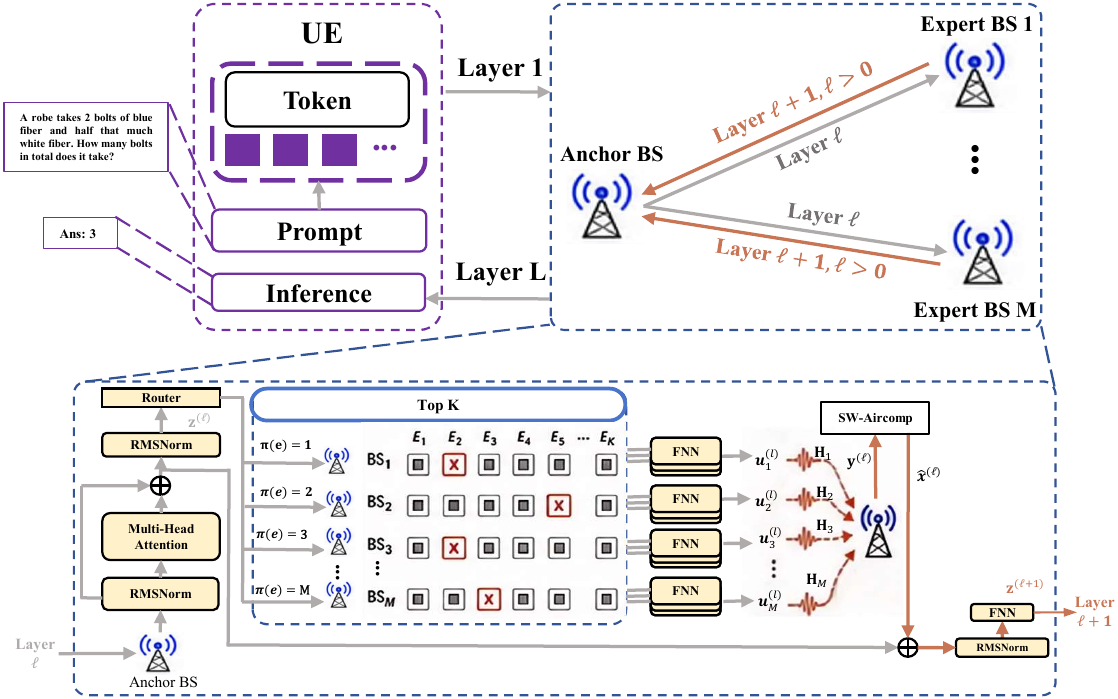}
        \centering
    \caption{System architecture of distributed MoE inference.}
    \label{fig:token-layer-flow}
\end{figure*}

The paper is organized as follows. Section~II presents the distributed MoE inference model and the MIMO AirComp signal model. Section~III develops the next-token logit-sensitivity model, derives the chunk-level sensitivity reduction, formulates the weighted transceiver problem, analyzes its convex block-coordinate structure, and presents the closed-loop algorithm. Section~IV evaluates the proposed method using Qwen3/OpenCompass replay, a direct ARC-Challenge closed-loop audit, and wireless sensitivity diagnostics. Section~V concludes the paper.

Notation: Bold lowercase letters denote vectors and bold uppercase letters denote matrices. The operators $(\cdot)^T$, $(\cdot)^H$, and $(\cdot)^*$ denote transpose, Hermitian transpose, and elementwise conjugate, respectively. The trace, rank, and diagonal operators are denoted by $\Tr(\cdot)$, $\rank(\cdot)$, and $\diag(\cdot)$. The expectation operator is $\E[\cdot]$. The notation $\mathbf X\succeq\bzero$ means that $\mathbf X$ is positive semidefinite. The identity matrix is $\bI$, and an all-zero vector or matrix is $\bzero$ with dimension clear from context. For complex vectors, $\Re\{\cdot\}$ and $\Im\{\cdot\}$ denote real and imaginary parts.

\section{System Model}\label{sec:system-model}

This section develops the system model for token-level
MoE aggregation at a given layer and decoding step within a
coordinated wireless edge network. We first describe the distributed
MoE inference architecture and expert placement in wireless communication network, and then formulate the corresponding MIMO AirComp model.

\subsection{Distributed MoE Inference Architecture}
\label{subsec:distributed-moe-aircomp-model}

We consider a pretrained LLM with $L$ MoE layers indexed by
$\ell\in\{1,\ldots,L\}$. The inference service is deployed over one
anchor BS and $M$ synchronized expert BSs. The anchor BS maintains the
decoding state, router, output head, and KV cache and executes the
non-expert model components, whereas the expert BSs store the subsets of
the pretrained MoE experts. Each expert BS is equipped with
$N_{\rm T}$ transmit antennas, while the anchor BS is equipped with
$N_{\rm R}$ receive antennas.

As illustrated in Fig.~\ref{fig:token-layer-flow}, a UE sends a prompt
to the anchor BS, where it is tokenized and processed autoregressively
by the LLM. At MoE layer $\ell$, the router at the anchor BS selects the
top-$K$ experts and dispatches the intermediate hidden state, together
with the corresponding routing weights, to the expert BSs hosting the
selected experts. Each active BS evaluates the selected expert or
experts stored locally and forms its router-weighted local contribution
by summing their weighted outputs. Because the subsequent LLM
computation requires only the aggregate of these local contributions,
the active BSs transmit them simultaneously over the wireless
multiple-access channel. By exploiting its natural signal
superposition, MIMO AirComp enables the anchor BS to recover the
weighted sum of the selected expert outputs without separately
decoding the contribution of each BS. The anchor BS then proceeds with
the remaining model computation. This procedure is repeated across the
$L$ MoE layers and throughout the autoregressive decoding steps. After
the final output token is generated, the anchor BS returns the
inference result to the UE. In the following, we formulate this
distributed MoE inference procedure as a frequency-domain MIMO AirComp
aggregation problem over wireless resource blocks.

\subsection{MIMO AirComp Model}

At MoE layer $\ell \in \mathcal{L}$, let $\bz^{(\ell)}\in\R^d$ be the hidden vector of the MoE block. Expert $e\in\cE^{(\ell)}$ computes
\begin{equation}
    \bv_{e}^{(\ell)}
    =
    E_e^{(\ell)}\!\left(\bz^{(\ell)}\right)
    \in\R^d,
    \label{eq:expert-output}
\end{equation}
where $E_e^{(\ell)}(\cdot)$ is a fixed pretrained expert network. The router selects a sparse expert subset $\cS^{(\ell)}\subseteq\cE^{(\ell)}$
and assigns weights $\{g_{e}^{(\ell)}\}_{e\in\cS^{(\ell)}}$ satisfying
$g_{e}^{(\ell)}\ge0$ and
$\sum_{e\in\cS^{(\ell)}}g_{e}^{(\ell)}=1$. Expert placement is defined by $\pi:\bigcup_{\ell\in\cL}\cE^{(\ell)}\rightarrow \cM$, where $\pi(e)=m$ means that BS $m$ contains the expert $e$; otherwise, $\pi(e) = 0$. The weighted local MoE contribution available at BS $m$ is $\bu_{m}^{(\ell)}=\sum_{e\in\cS^{(\ell)}:\,\pi(e)=m}g_{e}^{(\ell)}\bv_{e}^{(\ell)}\in\R^d$, with $\bu_{m}^{(\ell)}=\bzero$ if BS $m$ hosts no selected expert. The active BS set is defined as
\begin{equation}
    \cM^{(\ell)}
    =
    \left\{
    m\in\cM:
    \exists e\in\cS^{(\ell)}\ \text{such that}\ \pi(e)=m
    \right\}.
    \label{eq:active-bs-set}
\end{equation}
The ideal MoE aggregation output is $\bx^{(\ell+1)}=\sum_{m\in\cM^{(\ell)}}\bu_{m}^{(\ell)}$ and thus the AirComp target in the wireless aggregation design. Suppose that the channel from BS $m$ to the anchor BS can be modeled as a clustered frequency-domain MIMO channel. The channel matrix is given by
\begin{align}
    \bH_{m}
    =
    \sqrt{\frac{N_{\rm T}N_{\rm R}}{L_m}}
    \sum_{p=1}^{L_m}
    \alpha_{m,p}(\theta_{m,p}^{\rm T},\theta_{m,p}^{\rm R})
    \mathbf{a}_{R}(\theta_{m,p}^{\rm R})
    \mathbf{a}_{t}^{H}(\theta_{m,p}^{\rm T}),
    \label{eq:sim-channel-model}
\end{align}
where $L_m$ is the number of channel paths, $\alpha_{m,p}(\theta_{m,p}^{\rm T},\theta_{m,p}^{\rm R})$ is the path gain, $\theta_{m,p}^{\rm T}$ and $\theta_{m,p}^{\rm R}$ are the angle of departure (AoD) and angle of arrival (AoA), respectively. We assume uniform linear arrays (ULA), the steering vectors are defined as
\begin{align}
\mathbf{a}(\theta,N)
&=
\frac{1}{\sqrt{N_{\rm T}}}
\left[
1,
e^{j2\pi \frac{d}{\lambda}\sin(\theta)},
\ldots,
e^{j2\pi \frac{(N-1)d}{\lambda}\sin(\theta)}
\right]^T, 
\end{align}
where $d$ is the antenna spacing and $\lambda$ is the carrier wavelength. We assume quasi-static block fading within each MoE layer $\ell$, i.e., $\bH_{m,b}$ remains constant across information block $b$. The transmitted information is first mapped into a complex-domain representation through a linear and invertible packing operator $\mathcal{P}$. Specifically, BS $m$ maps its local MoE output $\mathbf{u}_{m}^{(\ell)}$ into
\begin{equation}
    \mathbf{s}_{m}^{(\ell)}
    =
    \mathcal{P}\left(\mathbf{u}_{m}^{(\ell)}\right),
\end{equation}
where $\mathbf{s}_{m,b}^{(\ell)}$ denotes the $b$-th information block of $\mathbf{s}_{m}^{(\ell)}$. Since $\mathcal{P}$ is linear, the desired block-wise AirComp target satisfies
\begin{equation}
    \sum_{m\in\mathcal{M}^{(\ell)}} \mathbf{s}_{m,b}^{(\ell)}
    =
    \left[
    \mathcal{P}\left(
    \sum_{m\in\mathcal{M}^{(\ell)}} \mathbf{u}_{m}^{(\ell)}
    \right)
    \right]_b ,
\end{equation}
where $b \in \mathcal{B}_{t,l}$ is the information block set at decoding step $t$ and MoE layer $l$. 
Therefore, the anchor BS is able to firstly estimate the summed complex-domain block $\sum_{m\in\mathcal{M}^{(\ell)}} \mathbf{s}_{m,b}^{(\ell)}$ through AirComp, and then applies $\mathcal{P}^{-1}$ after collecting all information blocks to recover the desired model-domain aggregation output $\mathbf{x}^{(\ell+1)}=\sum_{m\in\mathcal{M}^{(\ell)}}\mathbf{u}_{m}^{(\ell)}$.

During the wireless aggregation phase, each synchronized BS applies a linear precoder $\bF_m \in \C^{N_{\rm T}\times b}$ to its local MoE output $\bu_m^{(\ell)}$, and transmits $\bx_m^{(\ell)} = \bF_m \bs_m^{(\ell)}$. The received signal at the anchor BS is then given by the over-the-air superposition
\begin{equation}
    \by^{(\ell)}_{b}
    =
    \sum_{m\in\cM^{(\ell)}}
    \bH_{m,b} \bF_{m,b} \bs_{m,b}^{(\ell)}
    +
    \bw^{(\ell)},
    \label{eq:aircomp-received}
\end{equation}
where $\bw^{(\ell)} \sim \CN(\bzero,\sigma^2 \bI_{N_{\rm R}})$ denotes additive receiver noise. To recover the desired MoE aggregation output, the anchor BS applies a linear AirComp combiner $\bA^{(\ell)} \in \C^{N_{\rm R}\times d}$, producing
\begin{equation}
    \widehat{\bs}_{b}^{(\ell+1)}
    =
    \sum_{m\in\cM^{(\ell)}}
    \bA_{b}^{(\ell)H}\bH_{m,b} \bF_{m,b} \bs_{m,b}^{(\ell)}
    +
    \bA_{b}^{(\ell)H}\bw_{m,b}^{(\ell)},
    \label{eq:effective-aggregation}
\end{equation}
where $\tilde{\bw}_{m,b}^{(\ell)} =\bA_{b}^{(\ell)H}\bw_{{m,b}}^{(\ell)}$ is the effective noise. The ideal MoE aggregation $\bx^{(\ell+1)}=\sum_{m\in\cM^{(\ell)}}\bu_{m}^{(\ell)}$ is achieved if the following AirComp alignment condition holds:
$\bA_{b}^{(\ell)H}\bH_{m,b} \bF_{m,b} \approx \bI_d$, $\forall m \in \cM^{(\ell)}$. In this case, the received signal simplifies to
$\widehat{\bx}^{(\ell+1)}= \sum_{m}\widehat{\bu}_{m}^{(\ell+1)} = \mathcal{P}^{-1}\Big(\widehat{\bs}_{b}^{(\ell+1)}\Big)$. Hence, the resulting mismatch between the ideal MoE aggregation output and the AirComp reconstruction at decoding step $t$ is defined as $\bdelta_{t}^{(\ell)} =\widehat{\bx}_{t}^{(\ell+1)} - \bx_{t}^{(\ell+1)}$, the minimization of which is equivalent to the minimization of the information block difference, i.e., $\hat{\bs}^{(\ell)}_{b}-\bs^{(\ell)}_{b}$, where $\bs^{(\ell)}_{b} = \sum_{m\in\cM^{(\ell)}} \bs_{m,b}^{(\ell)}$. 

The transceiver design problem is therefore to jointly optimize the precoders $\{\bF_{m,b}\}$ and combiners $\bA_{b}^{(\ell)}$ such that the AirComp alignment condition $\bA_{b}^{(\ell)H}\bH_{m,b} \bF_{m,b} \approx \bI_d$, $\forall m \in \cM^{(\ell)}$ is approximately satisfied under the per-BS power constraint, which are formalized in the next section.

\section{Problem Formulation and Algorithm Design}
\label{sec:theory-algorithm}
This section develops a logit-aware AirComp design for sparse MoE aggregation over wireless edge networks. We first relate the wireless aggregation error to the next-token logit perturbation through the local logit Jacobian, which yields a block-level sensitivity weight that quantifies the average impact of each AirComp block on the LLM output. Using these weights, we formulate a logit-sensitivity-weighted transceiver optimization problem under per-BS total power constraints. Although the resulting problem is not convex, its receive-combiner and transmit-precoder subproblems admit exact block updates with standard convex structure. Finally, we present a compact closed-loop algorithm that estimates the sensitivity weights by finite differences, alternates between the closed-form receive and transmit updates, and yields a monotonic block-coordinate procedure for each fixed-weight AirComp objective.

\subsection{LLM-Sensitive Objective}
\label{subsec:llm-sensitive-objective}
At decoding step $t$, the LLM produces the logit vector $\bo_t=[o_{t,1},o_{t,2},\ldots,o_{t,V}]^T\in\R^V$. The clean next token can be obtained as $\argmax_{j\in\{1,\ldots,V\}}o_{t,j}$. Let $\bPhi_{\ell,t}:\R^d\rightarrow\R^V$ denote the mapping from the MoE
aggregation output at layer $\ell$ and decoding step $t$ to the next-token
logits. Under ideal and AirComp aggregation, the logit vectors are
\begin{align}
    \bo_t
    =
    \bPhi_{\ell,t}
    \!\left(
    \bx_t^{(\ell+1)}
    \right),~~~~~
    \widehat{\bo}_t
    =
    \bPhi_{\ell,t}
    \!\left(
    \hat{\bx}_t^{(\ell+1)}
    \right). \nonumber
\end{align}
The local logit Jacobian is
\begin{equation}
    \bJ_{\ell,t}
    =
    \frac{\partial\bPhi_{\ell,t}(\bx)}{\partial\bx}
    \bigg|_{\bx=\bx_t^{(\ell+1)}}
    \in\R^{V\times d}.
    \label{eq:logit-jacobian}
\end{equation}
For sufficiently small $\bdelta_t^{(\ell)}$,
\begin{equation}
    \widehat{\bo}_t-\bo_t
    =
    \bJ_{\ell,t}\bdelta_t^{(\ell)}
    +
    o\!\left(
    \left\|\bdelta_t^{(\ell)}\right\|_2
    \right).
    \label{eq:first-order-logit}
\end{equation}
Define the hidden-space sensitivity matrix $\bQ_{\ell,t}=\bJ_{\ell,t}^T\bJ_{\ell,t}\succeq\bzero$. Then the first-order squared logit distortion is
\begin{equation}
    \left\|
    \widehat{\bo}_t-\bo_t
    \right\|_2^2
    \approx
    \bdelta_t^{(\ell)T}
    \bQ_{\ell,t}
    \bdelta_t^{(\ell)}.
    \label{eq:first-order-logit-quadratic}
\end{equation}
Let $\bC_{\delta,t}^{(\ell)}=\E\!\left[\bdelta_t^{(\ell)}\bdelta_t^{(\ell)T}\right]$ be the hidden-error second moment. Taking expectation in \eqref{eq:first-order-logit-quadratic} yields
\begin{equation}
    \E
    \!\left[
    \left\|
    \widehat{\bo}_t-\bo_t
    \right\|_2^2
    \right]
    \approx
    \Tr
    \!\left(
    \bQ_{\ell,t}
    \bC_{\delta,t}^{(\ell)}
    \right).
    \label{eq:expected-logit-distortion-trace}
\end{equation}
Ordinary hidden-space MSE equals
$\E[\|\bdelta_t^{(\ell)}\|_2^2]=\Tr(\bC_{\delta,t}^{(\ell)})$. It is
aligned with \eqref{eq:expected-logit-distortion-trace} only when
$\bQ_{\ell,t}=\gamma\bI_d$ for some $\gamma>0$. In general, equal-energy
hidden errors can produce very different logit perturbations.

Let $\br_{b}\in\R^{q_{d}}$ denote the reconstruction error for block $b$. The expected norm of $\br_{b}$ is 
\begin{equation}
    \E\!\left[\|\br_{b}\|_2^2\right]
    =
    \E\!\left[\|\widehat{\bs}_{b}-\bs_{b}\|_2^2\right]
    \triangleq
    \epsilon_{b}.
    \label{eq:real-complex-error-normalization}
\end{equation}
Let $\mathbf G_{m,b}\in\R^{d\times b}$ embed this block error into the hidden space. Then
\begin{equation}
    \bdelta_t^{(\ell)}
    =
    \sum_{b\in\cB_{t,\ell}}
    \mathbf G_b\br_b.
    \label{eq:hidden-error-decomposition}
\end{equation}
Substituting \eqref{eq:hidden-error-decomposition} into
\eqref{eq:first-order-logit} gives
\begin{equation}
    \widehat{\bo}_t-\bo_t
    \approx
    \sum_{b\in\cB_{t,\ell}}
    \bJ_{\ell,t}\mathbf G_b\br_b.
    \label{eq:block-logit-error-sum}
\end{equation}
The block sensitivity matrix is
\begin{equation}
    \bOmega_b
    =
    \mathbf G_b^T
    \bQ_{\ell,t}
    \mathbf G_b
    =
    (\bJ_{\ell,t}\mathbf G_b)^T
    (\bJ_{\ell,t}\mathbf G_b)
    \succeq\bzero.
    \label{eq:block-sensitivity-matrix}
\end{equation}

If only block $b$ is perturbed, its first-order squared logit distortion is $\left\|\bJ_{\ell,t}\mathbf G_b\br_b\right\|_2^2=\br_b^T\bOmega_b\br_b$. Thus, $\bOmega_b$ is the restriction of $\bQ_{\ell,t}$ to the error subspace of block $b$. A full matrix-weighted design based on $\bOmega_b$ preserves directional
sensitivity inside each block but complicates the wireless transceiver. We
therefore use the scalar block weight
\begin{equation}
    \omega_b
    =
    \frac{1}{q_b}\Tr(\bOmega_b)
    =
    \frac{1}{q_b}
    \left\|
    \bJ_{\ell,t}\mathbf G_b
    \right\|_F^2.
    \label{eq:resource-weight}
\end{equation}
A larger $\omega_b$ means that the block-$b$ aggregation error has a larger
average first-order effect on the next-token logits.

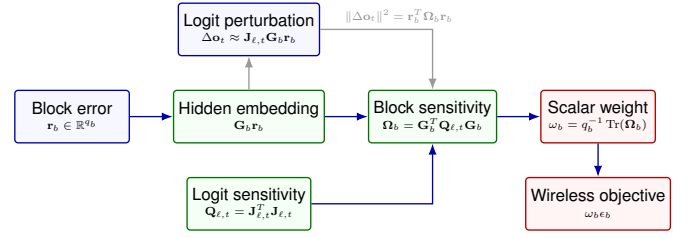
\begin{figure}[t]
    \centering
    \resizebox{0.99\columnwidth}{!}{
\begin{tikzpicture}[
    font=\sffamily,
    >=Latex,
    block/.style={draw=blue!65!black, thick, rounded corners=2pt, fill=blue!3, minimum width=24mm, minimum height=11mm, align=center, inner sep=3pt},
    mat/.style={draw=green!45!black, thick, rounded corners=2pt, fill=green!4, minimum width=24mm, minimum height=11mm, align=center, inner sep=3pt},
    weight/.style={draw=red!70!black, thick, rounded corners=2pt, fill=red!4, minimum width=24mm, minimum height=11mm, align=center, inner sep=3pt},
    arr/.style={->, thick, blue!65!black},
    smallarr/.style={->, thick, gray!80}
]
\node[block] (airerr) at (0,0) {Block error\\[-1mm]\scriptsize $\br_b\in\R^{q_b}$};
\node[mat, right=9mm of airerr] (recon) {Hidden embedding\\[-1mm]\scriptsize $\mathbf G_b\br_b$};
\node[mat, right=9mm of recon] (omega) {Block sensitivity\\[-1mm]\scriptsize $\bOmega_b=\mathbf G_b^T\bQ_{\ell,t}\mathbf G_b$};
\node[weight, right=9mm of omega] (w) {Scalar weight\\[-1mm]\scriptsize $\omega_b=q_b^{-1}\Tr(\bOmega_b)$};
\node[weight, below=7mm of w] (obj) {Wireless objective\\[-1mm]\scriptsize $\omega_b\epsilon_b$};
\node[mat, below=7mm of recon] (qmat) {Logit sensitivity\\[-1mm]\scriptsize $\bQ_{\ell,t}=\bJ_{\ell,t}^T\bJ_{\ell,t}$};
\node[block, above=7mm of recon] (logit) {Logit perturbation\\[-1mm]\scriptsize $\Delta\bo_t\approx\bJ_{\ell,t}\mathbf G_b\br_b$};

\draw[arr] (airerr) -- (recon);
\draw[arr] (recon) -- (omega);
\draw[arr] (qmat) -| (omega.south);
\draw[arr] (omega) -- (w);
\draw[arr] (w) -- (obj);
\draw[smallarr] (recon) -- (logit);
\draw[smallarr] (logit.east) -| node[pos=0.35, above, font=\scriptsize] {$\|\Delta\bo_t\|^2=\br_b^T\bOmega_b\br_b$} (omega.north);
\end{tikzpicture}}
    \caption{Construction of scalar AirComp weights from block-level PHY error and local next-token logit sensitivity.}
    \label{fig:sensitivity-pipeline}
\end{figure}

\begin{lemma}
\label{lem:scalar-sensitivity-reduction}
For a fixed pair $(t,\ell)$, suppose
\eqref{eq:hidden-error-decomposition} holds. If
$\{\br_b\}_{b\in\cB_{t,\ell}}$ are mutually uncorrelated and satisfy
\begin{equation}
    \E\!\left[\br_b\br_b^T\right]
    =
    \frac{\epsilon_b}{q_b}\bI_{q_b},
    \qquad b\in\cB_{t,\ell},
    \label{eq:isotropic-resource-error}
\end{equation}
then, under the first-order logit approximation,
\begin{equation}
    \E\!\left[
    \left\|
    \widehat{\bo}_t-\bo_t
    \right\|_2^2
    \right]
    \approx
    \sum_{b\in\cB_{t,\ell}}
    \omega_b\epsilon_b.
    \label{eq:weighted-mse-reduction}
\end{equation}
\begin{proof}
    See Appendix~\ref{app:proof-scalar-sensitivity}.
\end{proof}
\end{lemma}

\begin{theorem}
\label{thm:sensitivity-bound}
Suppose $\bPhi_{\ell,t}$ has a locally $L_{\ell,t}$-Lipschitz Jacobian in a
neighborhood of $\bx_t^{(\ell+1)}$ and
\eqref{eq:hidden-error-decomposition} holds for a fixed pair $(t,\ell)$.
If the block errors $\{\br_b\}_{b\in\cB_{t,\ell}}$ are mutually uncorrelated
and satisfy \eqref{eq:isotropic-resource-error}, then, for any $\rho>0$,
\small
\begin{align}
    \E\!\left[
    \left\|
    \widehat{\bo}_t-\bo_t
    \right\|_2^2
    \right]
    \le&
    (1+\rho)
    \sum_{b\in\cB_{t,\ell}}
    \omega_b\epsilon_b
   +
    \frac{1+\rho^{-1}}{4}
    L_{\ell,t}^2
    \E\!\left[
    \left\|
    \bdelta_t^{(\ell)}
    \right\|_2^4
    \right].
    \label{eq:sensitivity-weighted-upper-bound}
\end{align}
\begin{proof}
    See Appendix~\ref{app:proof-sensitivity-bound}.
\end{proof}
\end{theorem}

\begin{corollary}
\label{cor:token-preservation}
Let $y_t^\star=\argmax_j o_{t,j}$ be the clean greedy token and define the
clean top-token logit gap
$\Gamma_t=o_{t,y_t^\star}-\max_{j\ne y_t^\star}o_{t,j}$. If
\begin{equation}
    \left\|
    \widehat{\bo}_t-\bo_t
    \right\|_\infty
    <
    \frac{\Gamma_t}{2},
    \label{eq:token-preservation-condition}
\end{equation}
then greedy decoding is unchanged. Moreover, if
$D_t\ge\E[\|\widehat{\bo}_t-\bo_t\|_2^2]$, then
\begin{equation}
    \Pr\!\left[
    \argmax_j\widehat{o}_{t,j}
    \ne
    y_t^\star
    \right]
    \le
    \frac{4D_t}{\Gamma_t^2}.
    \label{eq:token-flip-bound}
\end{equation}
The value $D_t$ may be chosen as the right-hand side of
\eqref{eq:sensitivity-weighted-upper-bound}.
\begin{proof}
    See Appendix~\ref{app:proof-token-preservation}.
\end{proof}
\end{corollary}

Corollary~\ref{cor:token-preservation} is used only to interpret the
logit-aware objective. It shows that the top-token decision is stable when the
aggregation-induced logit perturbation is small compared with the clean
top-token logit gap. The optimizer itself does not introduce an additional
gap-dependent weight. Throughout the transceiver design, the only block
priority is the logit-sensitivity weight $\omega_b$ in
\eqref{eq:resource-weight}; its practical estimate is denoted by
$\widehat{\omega}_b$ in \eqref{eq:finite-diff-weight}.

\subsection{Logit-Sensitivity-Weighted MIMO AirComp Design}
\label{subsec:sensitivity-weighted-design}

For the current optimization window, denote the active AirComp
block set $\cB_m=\{b\in\cB:m\in\cM_b\}$ for the blocks served by BS $m$. Assume that $\{\bs_{m,b}\}_{m\in\cM_b}$ are zero mean and mutually
uncorrelated for each $b$, and that the receiver noise is independent of the
sources. The source covariance is
\begin{equation}
    \bR_{m,b}
    =
    \E\!\left[
    \bs_{m,b}\bs_{m,b}^{H}
    \right]
    \succeq\bzero.
    \label{eq:source-covariance}
\end{equation}
The AirComp MSE on block $b$ is
\begin{equation}
    \epsilon_b
    =
    \E\!\left[
    \left\|
    \bA_b^H
    \left(
    \sum_{m\in\cM_b}
    \bH_{m,b}\bF_{m,b}\bs_{m,b}
    +
    \bw_b
    \right)
    -
    \bs_b
    \right\|_2^2
    \right].
    \label{eq:block-mse}
\end{equation}
Therefore, the sensitivity-weighted AirComp design is to minimize the weighted-sum MSE over all blocks \begin{subequations}
\label{prob:main}
\begin{align}
    \underset{\{\bF_{m,b}\},\{\bA_b\}}{\operatorname{minimize}}
    \quad&
    J(\bF,\bA)
    \triangleq
    \sum_{b\in\cB}
    \omega_b\epsilon_b
    \label{prob:main-obj}\\
    \operatorname{subject\;to}
    \quad&
    \sum_{b\in\cB_m}
    \Tr\!\left(
    \bF_{m,b}\bR_{m,b}\bF_{m,b}^H
    \right)
    \le
    P_m^{\rm tot},
    \quad
    \forall m\in\cM .
    \label{prob:main-total-power}
\end{align}
\end{subequations}
The total-power constraint in \eqref{prob:main-total-power} couples all
AirComp blocks served by the same BS. Therefore, the logit-sensitivity weights
change the allocation of transmit power across blocks: a more sensitive block
must compete with the other blocks for the same BS-level power budget. This
is the minimal coupling required for the weights $\{\omega_b\}$ to affect the
transceiver design. Optional implementation constraints such as spectral masks
or hardware limits can be added separately, but they are not part of the
core optimization problem. Problem~\eqref{prob:main} is not jointly convex
because $\bA_b^H\bH_{m,b}\bF_{m,b}$ is bilinear in the combiner and precoder
variables. The following lemma shows that its exact block updates have
standard convex structure.

\begin{lemma}
\label{lem:mse-decomposition}
The block MSE in \eqref{eq:block-mse} satisfies
\begin{equation}
    \epsilon_b
    =
    \sum_{m\in\cM_b}
    \left\|
    \left(
    \bA_b^H\bH_{m,b}\bF_{m,b}
    -
    \bI_r
    \right)
    \bR_{m,b}^{1/2}
    \right\|_F^2
    +
    \sigma_b^2
    \Tr\!\left(
    \bA_b^H\bA_b
    \right).
    \label{eq:mse-decomposition}
\end{equation}
For fixed transmit precoders, define
\begin{align}
    &\bR_{y,b}
    =
    \E[\by_b\by_b^H]
    =
    \sum_{m\in\cM_b}
    \bH_{m,b}\bF_{m,b}
    \bR_{m,b}
    \bF_{m,b}^H\bH_{m,b}^H
    +
    \sigma_b^2\bI_{N_{\rm R}},
    \label{eq:Ry}\\
    &\bR_{ys,b}
    =
    \E[\by_b\bs_b^H]
    =
    \sum_{m\in\cM_b}
    \bH_{m,b}\bF_{m,b}\bR_{m,b}.
    \label{eq:Rys}
\end{align}
Since $\sigma_b^2>0$, $\bR_{y,b}\succ\bzero$, the fixed-precoder
receive-combiner block of \eqref{prob:main} decouples over $b$ as
\begin{equation}
    \underset{\bA_b}{\operatorname{minimize}}
    \quad
    \omega_b
    \left[
    \Tr\!\left(
    \bA_b^H\bR_{y,b}\bA_b
    \right)
    -
    2\Re\!\left\{
    \Tr\!\left(
    \bA_b^H\bR_{ys,b}
    \right)
    \right\}
    \right],
    \label{prob:rx-block}
\end{equation}
and its unique minimizer is
\begin{equation}
    \bA_b^\star
    =
    \bR_{y,b}^{-1}\bR_{ys,b}.
    \label{eq:optimal-combiner}
\end{equation}
\begin{proof}
    See Appendix~\ref{app:proof-mse-decomposition}.
\end{proof}
\end{lemma}

Lemma~\ref{lem:mse-decomposition} connects the
physical AirComp MSE to the receive-combiner update: the matrices
$\bR_{y,b}$ and $\bR_{ys,b}$ are precisely the quadratic and linear
coefficients of the fixed-precoder MIMO combiner problem
\eqref{prob:rx-block}. Thus, the receive update is a standard linear-MMSE
combiner for estimating the aggregate $\bs_b$, not the individual user
symbols.

For fixed receive combiners, the BS-$m$ transmit subproblem is
\begin{subequations}
\label{prob:tx-block}
\begin{align}
    \underset{\{\bF_{m,b}\}_{b\in\cB_m}}{\operatorname{minimize}}
    \quad&
    \sum_{b\in\cB_m}
    \omega_b
    \left\|
    \left(
    \bA_b^H\bH_{m,b}\bF_{m,b}
    -
    \bI_r
    \right)
    \bR_{m,b}^{1/2}
    \right\|_F^2
    \label{prob:tx-block-obj}\\
    \operatorname{subject\;to}
    \quad&
    \sum_{b\in\cB_m}
    \Tr\!\left(
    \bF_{m,b}\bR_{m,b}\bF_{m,b}^H
    \right)
    \le
    P_m^{\rm tot}.
    \label{prob:tx-block-total}
\end{align}
\end{subequations}

\begin{proposition}[]
\label{prop:precoder}
For fixed $\{\bA_b\}$ and $\omega_b>0$, the BS-$m$ transmit subproblem
\eqref{prob:tx-block} is a convex quadratically constrained quadratic program (QCQP). Its globally optimal update is
\begin{align}
    &\bF_{m,b}^{\star}(\lambda_m)
    =
    \left(
    \bC_{m,b}
    +
    \frac{\lambda_m}{\omega_b}\bI_{N_{\rm T}}
    \right)^\dagger
    \bH_{m,b}^H\bA_b, b\in\cB_m,
    \label{eq:optimal-precoder-weighted}\end{align}
where the scalar dual variable is selected as
\begin{equation}
    \lambda_m^\star
    =
    \begin{cases}
    0,
    & P_m(0)\le P_m^{\rm tot},\\
    \text{the solution of } P_m(\lambda_m)=P_m^{\rm tot},
    & P_m(0)>P_m^{\rm tot}.
    \end{cases}
    \label{eq:precoder-dual-selection}
\end{equation} with $\bC_{m,b}=\bH_{m,b}^H\bA_b\bA_b^H\bH_{m,b}$ and 
$P_m(\lambda_m)
    =
    \sum_{b\in\cB_m}
    \Tr\!\left(
    \bF_{m,b}^{\star}(\lambda_m)
    \bR_{m,b}
    \bF_{m,b}^{\star H}(\lambda_m)
    \right)$.
\begin{proof}
    See Appendix~\ref{app:proof-precoder}.
\end{proof}
\end{proposition}
The weight $\omega_b$ sets the importance of block-$b$ AirComp error in
the objective \eqref{prob:main-obj}. Reducing $\epsilon_b$ by
$\Delta\epsilon_b$ decreases the objective by
$\omega_b\Delta\epsilon_b$. Hence, under the shared BS power budget
\eqref{prob:main-total-power}, the optimizer gives higher priority to
blocks whose errors have larger predicted logit impact. This can also be
seen from \eqref{eq:optimal-precoder-weighted}: if the precoder obtained
with $\lambda_m=0$ would exceed the BS power budget, then $\lambda_m$ is
increased until the total transmit power equals $P_m^{\rm tot}$. In that
case, a larger $\omega_b$ gives a smaller ratio $\lambda_m/\omega_b$ for
block $b$. If the precoder obtained with $\lambda_m=0$ already satisfies
the BS power budget, then $\lambda_m=0$ and the weight does not change the
precoder. If each block has its own fixed power budget, multiplying that
single-block objective by the positive scalar $\omega_b$ would not change
the block minimizer.

\subsection{Compact Closed-Loop Algorithm}
\label{subsec:compact-algorithm}

The only parameter that must be estimated from the LLM is the sensitivity
weight $w_b$ in \eqref{eq:resource-weight}. Directly evaluating $w_b$ requires the full Jacobian
$\bJ_{\ell,t}\in\R^{V\times d}$; instead, we estimate the same trace value by
finite differences.

For each block $b$ selected for weight estimation, draw independent random
vectors $\{\bq_{b,s}\}_{s=1}^{S}\subset\R^{q_b}$ satisfying
$\E[\bq_{b,s}\bq_{b,s}^T]=\bI_{q_b}$. For perturbation radius
$\varepsilon>0$, define
\begin{equation}
    \widehat{\omega}_b
    =
    \frac{1}{q_bS}
    \sum_{s=1}^{S}
    \frac{
    \left\|
    \bPhi_{\ell,t}
    \!\left(
    \bx_t^{(\ell+1)}
    +
    \varepsilon\mathbf G_b\bq_{b,s}
    \right)
    -
    \bPhi_{\ell,t}
    \!\left(
    \bx_t^{(\ell+1)}
    \right)
    \right\|_2^2
    }{\varepsilon^2}.
    \label{eq:finite-diff-weight}
\end{equation}
Thus, $\widehat{\omega}_b$ is not a second kind of weight; it is the
Monte-Carlo finite-difference estimate of the sensitivity weight
$\omega_b$. For small $\varepsilon$, a first-order expansion of
$\bPhi_{\ell,t}$ gives
$\E[\widehat{\omega}_b]\approx q_b^{-1}\Tr(\bOmega_b)=\omega_b$.

If the logit sensitivity is recomputed over several calibration rounds, the
stored base weight is updated as
\begin{equation}
    \omega_b^{(q+1)}
    =
    \operatorname{clip}
    \!\left(
    (1-\alpha)\omega_b^{(q)}
    +
    \alpha\widehat{\omega}_b^{(q)},
    \omega_{\min},
    \omega_{\max}
    \right),
    \label{eq:sensitivity-tracking}
\end{equation}
where $q$ is the weight-estimation round, $\alpha\in(0,1]$ is the averaging
stepsize, and clipping enforces
$0<\omega_{\min}\le\omega_b\le\omega_{\max}$. If the weights are computed only
once, one can set $\alpha=1$ and use the clipped estimate from
\eqref{eq:finite-diff-weight} directly in the fixed-weight inner loop.

\begin{algorithm}[t]
\caption{Logit-Sensitivity-Weighted MIMO AirComp}
\label{alg:llm-sensitive-aircomp}
\begin{algorithmic}[1]
\Require $\{\bH_{m,b},\bR_{m,b},\mathbf G_b\}$, $\{P_m^{\rm tot}\}$, $S$, $\varepsilon$, $\alpha$, $\omega_{\min}$, $\omega_{\max}$, $\tau_{\rm in}$, $\tau_{\rm out}$.
\Ensure $\{\omega_b,\bA_b^\star,\bF_{m,b}^\star\}$.
\State Initialize feasible $\{\bF_{m,b}\}$ and $\omega_b>0$, $\forall b\in\cB$, $q=1$.
\While{$\text{true}$}
    \State Estimate $\widehat{\omega}_b^{(q)}$ by \eqref{eq:finite-diff-weight}; update $\omega_b^{(q+1)}$ by \eqref{eq:sensitivity-tracking}, $\forall b\in\cB$.
    \State Set $\bF^{(0)}$ to the current feasible precoders; set $i\leftarrow0$ and $J^{(0)}\leftarrow+\infty$.
    \Repeat
        \State $\bA_b^{(i+1)}\leftarrow\eqref{eq:optimal-combiner}$ using \eqref{eq:Ry}--\eqref{eq:Rys}, $\forall b\in\cB$.
        \State Find $\lambda_m^{(i+1)}$ by \eqref{eq:precoder-dual-selection}, $\forall m\in\cM$.
        \State $\bF_{m,b}^{(i+1)}\leftarrow\bF_{m,b}^{\star}(\lambda_m^{(i+1)})$ in \eqref{eq:optimal-precoder-weighted}, $\forall m\in\cM,b\in\cB_m$.
        \State $J^{(i+1)}\leftarrow J(\bF^{(i+1)},\bA^{(i+1)})$ in \eqref{prob:main-obj}; $i\leftarrow i+1$.
    \Until{$i\ge2$ and $|J^{(i)}-J^{(i-1)}|/\max\{1,J^{(i-1)}\}\le\tau_{\rm in}$}
    \If{$\|\boldsymbol{\omega}^{(q+1)}-\boldsymbol{\omega}^{(q)}\|_\infty/\max\{1,\|\boldsymbol{\omega}^{(q)}\|_\infty\}\le\tau_{\rm out}$}
        \State \textbf{break}
    \EndIf
    \State $q = q +1$.
\EndWhile
\end{algorithmic}
\end{algorithm}

\begin{proposition}[Monotonic convergence of the fixed-weight inner loop]
\label{prop:inner-loop-convergence}
Fix $\{\omega_b\}$, $\{\bH_{m,b}\}$, and $\{\bR_{m,b}\}$ with bounded
positive weights. From any feasible finite-power precoders, run the inner loop of
Algorithm~\ref{alg:llm-sensitive-aircomp} with the exact receive update
\eqref{eq:optimal-combiner} and transmit update
\eqref{eq:optimal-precoder-weighted}--\eqref{eq:precoder-dual-selection}.
Then the generated objective values satisfy
\begin{equation}
    J(\bF^{(i+1)},\bA^{(i+1)})
    \le
    J(\bF^{(i)},\bA^{(i)}).
    \label{eq:inner-loop-descent}
\end{equation}
Hence the fixed-weight objective sequence converges. Every accumulation
point is block-stationary for \eqref{prob:main} if it exists.
\end{proposition}

Proposition~\ref{prop:inner-loop-convergence} is a fixed-weight statement.
The outer loop refreshes the finite-difference sensitivity weights and thus
tracks a changing objective; monotonicity applies after those weights are fixed.

\begin{table*}[t]
    \centering
    \caption{Examples of OpenCompass Subsets}
    \label{tab:opencompass-examples}
    \scriptsize
    \begin{tabularx}{\textwidth}{@{}p{0.12\textwidth}p{0.07\textwidth}X p{0.16\textwidth}@{}}
        \toprule
        Benchmark & Item & Prompt excerpt and choices & Gold answer \\
        \midrule
        GSM8K & \#0 & Janet's ducks lay 16 eggs per day. She eats three for breakfast and uses four for muffins. She sells the remainder at the farmers' market for \$2 each. & 18 \\
        GSM8K & \#1 & A robe takes 2 bolts of blue fiber and half that much white fiber. How many bolts in total does it take? & 3 \\
        ARC-c & \#0 & Juan and LaKeisha roll a few objects down a ramp. They want to see which object rolls the farthest. What should they do so they can repeat their investigation? Choices: A) Put the objects in groups; B) Change the height of the ramp; C) Choose different objects to roll; D) Record the details of the investigation. & D \\
        ARC-c & \#1 & High-pressure systems stop air from rising into the colder regions of the atmosphere where water can condense. What will most likely result if a high-pressure system remains in an area for a long period of time? Choices: A) fog; B) rain; C) drought; D) tornado. & C \\
        \bottomrule
    \end{tabularx}
\end{table*}

Let $B=|\cB|$ and let $I_{\rm in}$ be the number of inner block coordinate descent (BCD) iterations.
Each receive update requires one $N_{\rm R}\times N_{\rm R}$ linear solve.
Each transmit update requires one $N_{\rm T}\times N_{\rm T}$ linear solve
per active BS-block pair, plus scalar search over $\lambda_m$ to enforce
\eqref{prob:main-total-power}. With dual-search tolerance $\tau_\lambda$,
the per-round computational complexity is
\begin{equation}
    \mathcal O
    \!\left(
    I_{\rm in}
    B
    \left[
    N_{\rm R}^3
    +
    M_{\rm act}N_{\rm T}^3\log(1/\tau_\lambda)
    \right]
    \right),
    \label{eq:complexity}
\end{equation}
where $M_{\rm act}$ is the average number of active BSs per aggregation
block. The finite-difference calibration in \eqref{eq:finite-diff-weight}
requires $S$ perturbed logit evaluations per calibrated block. This overhead
can be reduced by sharing weights across nearby blocks, calibrating only
high-traffic layers, or replacing finite differences with Jacobian-vector
and vector-Jacobian products when model internals are accessible.

\section{Simulation Results}
\label{sec:simulation-design}

This section evaluates the proposed framework from three complementary angles.
First, a controlled Qwen3/OpenCompass replay tests whether errors injected into the
MoE aggregation outputs can change task accuracy. Second, a direct
ARC-Challenge closed-loop audit inserts AirComp aggregation perturbations into
Qwen3 token generation and measures task accuracy. Third, a
MIMO wireless simulator tests whether the proposed sensitivity-weighted MIMO AirComp design reduces the aggregation error that the LLM is most sensitive to. The first experiment is a pretrained-model sensitivity check, the second is a task-level AirComp comparison on ARC-Challenge, and the third is a physical-layer optimizer comparison. The results show that MoE aggregation quality is important to the inference quality, where the sensitivity-weighted AirComp (SW-AirComp) improves direct ARC-Challenge accuracy relative to the tested AirComp baselines, and the proposed transceiver effectively reduces LLM-relevant wireless distortion.
\subsection{Benchmark Setup}
\label{subsec:real-moe-benchmark-setup}

OpenCompass is used as the benchmark and scoring framework, while GSM8K and
ARC-Challenge provide the evaluated datasets~\cite{opencompass2023}.
Table~\ref{tab:opencompass-examples} lists representative examples from the
evaluated subsets. The pretrained model used in the replay is
Qwen3-30B-A3B-Instruct-2507-FP8, which has 48 MoE layers and 128 routed
experts in the local implementation. The clean reference uses the original
pretrained forward pass. The perturbed replay uses the same checkpoint,
prompts, tokenizer, and OpenCompass settings, but replaces each selected MoE
aggregation output by
$\widehat{\bx}_t^{(\ell+1)}=\bx_t^{(\ell+1)}+\bdelta_t^{(\ell)}$.
This controlled perturbation emulates aggregation error at the hidden-state
interface and makes the benchmark sensitive to MoE aggregation quality. For a
target aggregation SNR, the perturbation energy is normalized to the clean MoE
aggregate energy as follows:
\begin{equation}
    \E\!\left[
    \frac{\|\bdelta_t^{(\ell)}\|_2^2}{d_h}
    \right]
    =
    \left(
    \frac{\|\bx_t^{(\ell+1)}\|_2^2}{d_h}
    \right)
    10^{-{\rm SNR}/10},
\end{equation}
where $\bx_t^{(\ell+1)}$ is the clean MoE aggregate, $d_h$ is the hidden
dimension, and $\bdelta_t^{(\ell)}$ is independent Gaussian aggregation noise.
GSM8K is evaluated with OpenCompass zero-shot generation on a fixed
500-example subset. ARC-Challenge is evaluated as a multiple-choice task, where
the model selects one answer from the candidate options and the result is
scored by accuracy. The tested perturbation SNRs span $0$--$30$ dB.
\begin{figure}[t]
    \centering
 \includegraphics[height=0.5\linewidth,width=0.8\linewidth]{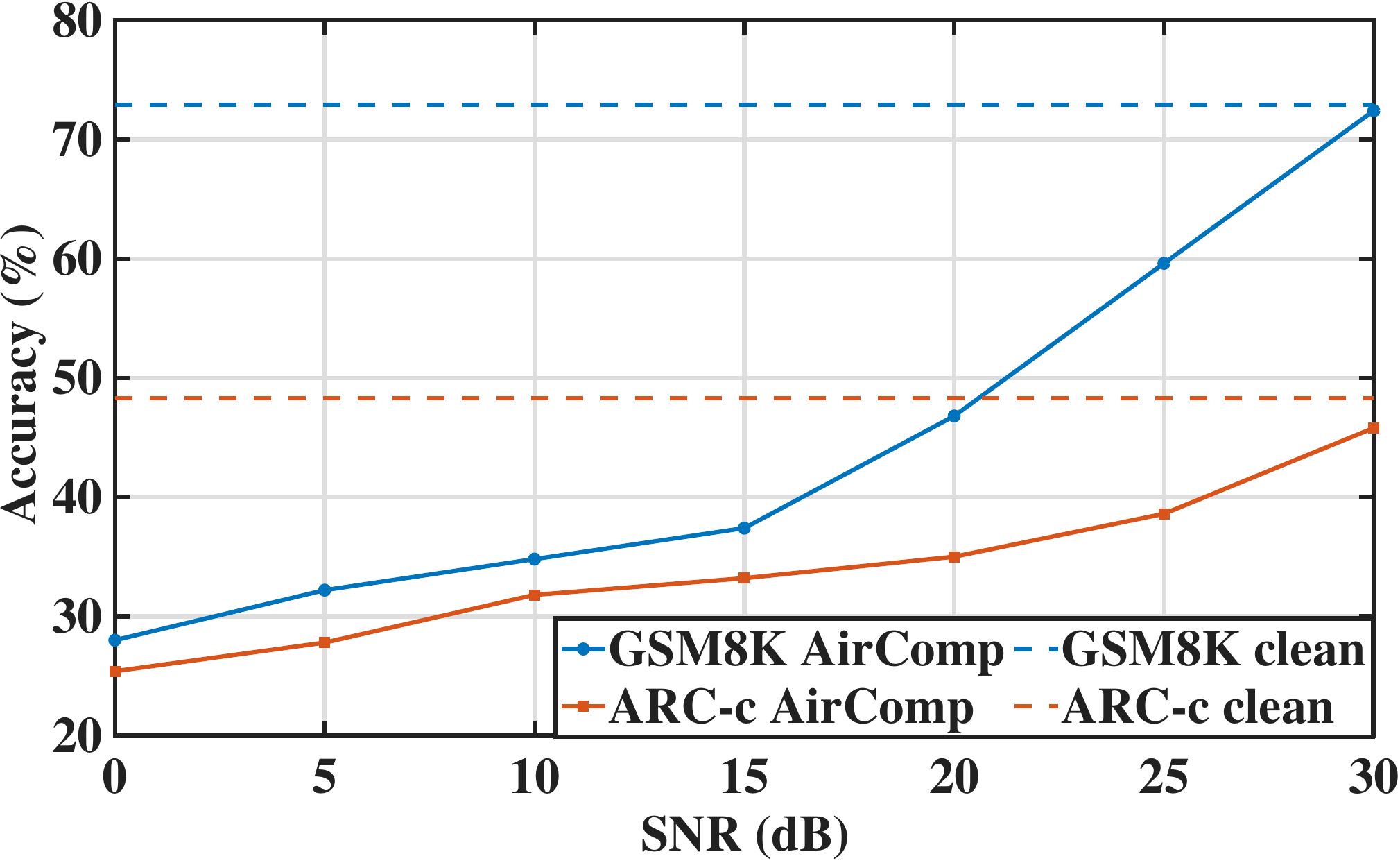}
        \caption{Qwen3 OpenCompass accuracy under controlled MoE aggregation perturbations.}
        \label{fig:qwen3-opencompass-snr}
\end{figure}
\begin{table}[!t]
    \centering
    \caption{Wireless Parameters}
    \label{tab:wireless-sim-setup}
    \scriptsize
    \begin{tabular}{@{}lc@{}}
        \toprule
        Parameter & Value \\
        \midrule
        Number of BSs & $M=4$ \\
        Active BSs per aggregation & $|\cM_b|=2$ \\
        Transmit antennas per BS & $N_{\rm T}\in\{32,64,128\}$ \\
        Receive antennas & $N_{\rm R}=2$ \\
        System bandwidth & $20$ MHz \\
        Maximum BS transmit power & $P_m^{\rm tot}=1$ W ($30$ dBm) \\
        \bottomrule
    \end{tabular}
\end{table}

\subsection{Wireless Simulator}
\label{subsec:wireless-sim-metrics}
The wireless simulator instantiates 16 aggregation blocks. These blocks are
the wireless optimization units used to pack hidden-state coordinates and
block-level sensitivities into resource-block MIMO transmissions; they should not be
read as the number of MoE layers in Qwen3. The simulator evaluates AirComp
performance under controlled channels, SNRs, and antenna configurations. It
uses the sensitivity weights derived from hidden/logit traces and clean
top-token logit gaps to test whether the wireless optimizer allocates protection
to the aggregation blocks that are most important for next-token preservation.

Table~\ref{tab:wireless-sim-setup} summarizes the wireless simulation parameters. The simulator considers a MIMO aggregation channel over frequency-domain resource blocks with system bandwidth of $20$ MHz. The transmit power follows the per-BS power constraint in \eqref{prob:main-total-power}. In the performance evaluation, The compared practical schemes are SW-AirComp, UW-AirComp,
MF-AirComp, and ZF-AirComp. MF-AirComp denotes a matched-filter AirComp
baseline. SW-AirComp
denotes the proposed logit-sensitivity-weighted MMSE AirComp design. It minimizes the
logit-sensitivity-weighted aggregation MSE, so that MoE aggregation variables with
larger impact on the final inference output receive higher protection in the
AirComp transceiver optimization. UW-AirComp is the unweighted MMSE AirComp
baseline, where the same optimization is performed with all sensitivity weights
set to one. It uses channel-matched analog transmit and receive beamforming,
followed by power normalization to satisfy the per-BS transmit-power
constraint. ZF-AirComp denotes a zero-forcing AirComp baseline, which attempts
to invert the effective MIMO channel so that the received superposition matches
the desired aggregate. MF-AirComp and ZF-AirComp do not use the logit-sensitivity
weights, and therefore serve as PHY-layer AirComp baselines. These
baselines separate the gain from logit-sensitivity weighting from the gain of the
proposed AirComp transceiver optimization.

\subsection{Qwen3 OpenCompass Replay}
\label{subsec:qwen3-opencompass-results}
Fig.~\ref{fig:qwen3-opencompass-snr} reports the Qwen3 inference accuracy under perturbations. Fig.~\ref{fig:qwen3-opencompass-snr} combines GSM8K and ARC-Challenge because
both measurements use the same checkpoint, perturbation replay path, SNR,
and accuracy metric. GSM8K shows the clearest task-level effect. For GSM8K, the replay accuracy increases from approximately $28.0\%$ at $0$ dB to $46.8\%$ at $20$ dB and $72.4\%$ at $30$ dB,
approaching the approximately $72.9\%$ clean-reference accuracy. These results demonstrate that aggregation errors can substantially
degrade inference performance, whereas sufficiently
reliable aggregation recovers nearly all of the clean-model accuracy.

ARC-Challenge provides a harder non-math check. The replay reaches
$45.8\%$ at 30 dB compared with the $48.3\%$ clean reference, while the
0 dB audit gives $26.44\%$ accuracy averaged over ten noise seeds. The ARC-Challenge confirms that MoE aggregation
perturbations affect downstream multiple-choice accuracy and that high SNR is
better in this run. Based on the plotted results, the clean-normalized retention at
30 dB is approximately $72.4/72.9=99.3\%$ for GSM8K and
$45.8/48.3=94.8\%$ for ARC-Challenge.

To close the task-level loop for the AirComp baselines, we further run a direct
ARC-Challenge audit in which calibrated aggregation perturbations from each
AirComp design are inserted into every Qwen3 MoE aggregation call during
autoregressive decoding. The hooks cover all 48 Qwen3 MoE layers in the local
implementation, and the generated answers are scored over 295 ARC-Challenge
examples using the same answer-extraction pipeline. Fig.~\ref{fig:arc-direct-closed-loop}
reports the resulting closed-loop accuracy. At 30~dB, SW-AirComp obtains
$43.73\%$ accuracy ($129/295$), compared with $40.34\%$ for UW-AirComp,
$37.29\%$ for MF-AirComp, and $33.56\%$ for ZF-AirComp. At 20~dB,
SW-AirComp obtains $26.78\%$ accuracy ($79/295$), compared with $25.42\%$,
$24.41\%$, and $24.07\%$ for the same baselines. Thus, the direct ARC-Challenge
closed-loop result supports the stronger task-level statement: for the
tested ARC-Challenge subset and high-SNR operating points, the
SW-AirComp perturbation path yields higher Qwen3 task
accuracy than AirComp baselines.

\begin{figure}[t]
    \centering
    \safeincludegraphics[width=0.75\linewidth]{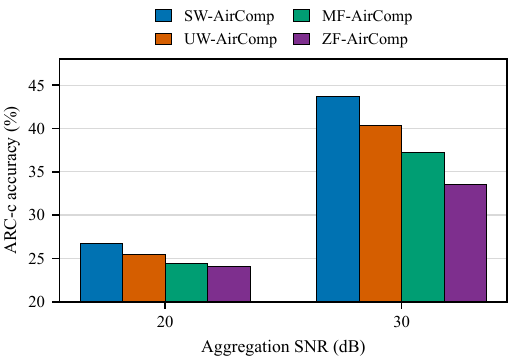}
    \caption{ARC-Challenge Accuracy: AirComp Perturbations.}
    \label{fig:arc-direct-closed-loop}
\end{figure}

The simulator suggests three operating regimes for distributed MoE AirComp. At
very low SNR, all baselines are dominated by insufficient received signal
quality, so fine-grained sensitivity allocation cannot compensate for the loss
of the hidden state aggregation. At very high SNR, residual distortion is small for all
reasonable transceiver designs, and the gap among methods should narrow. The
most important region is the middle SNR range, where the system has a limited
but useful error budget. In this regime, the transceiver cannot remove every
aggregation error, but it can decide where the residual error should be placed.

\subsection{Massive-MIMO Scaling Study}
\label{subsec:massive-mimo-results}

\begin{figure}[t]
    \centering
        \includegraphics[width=0.75\linewidth,height = 0.24\textwidth]{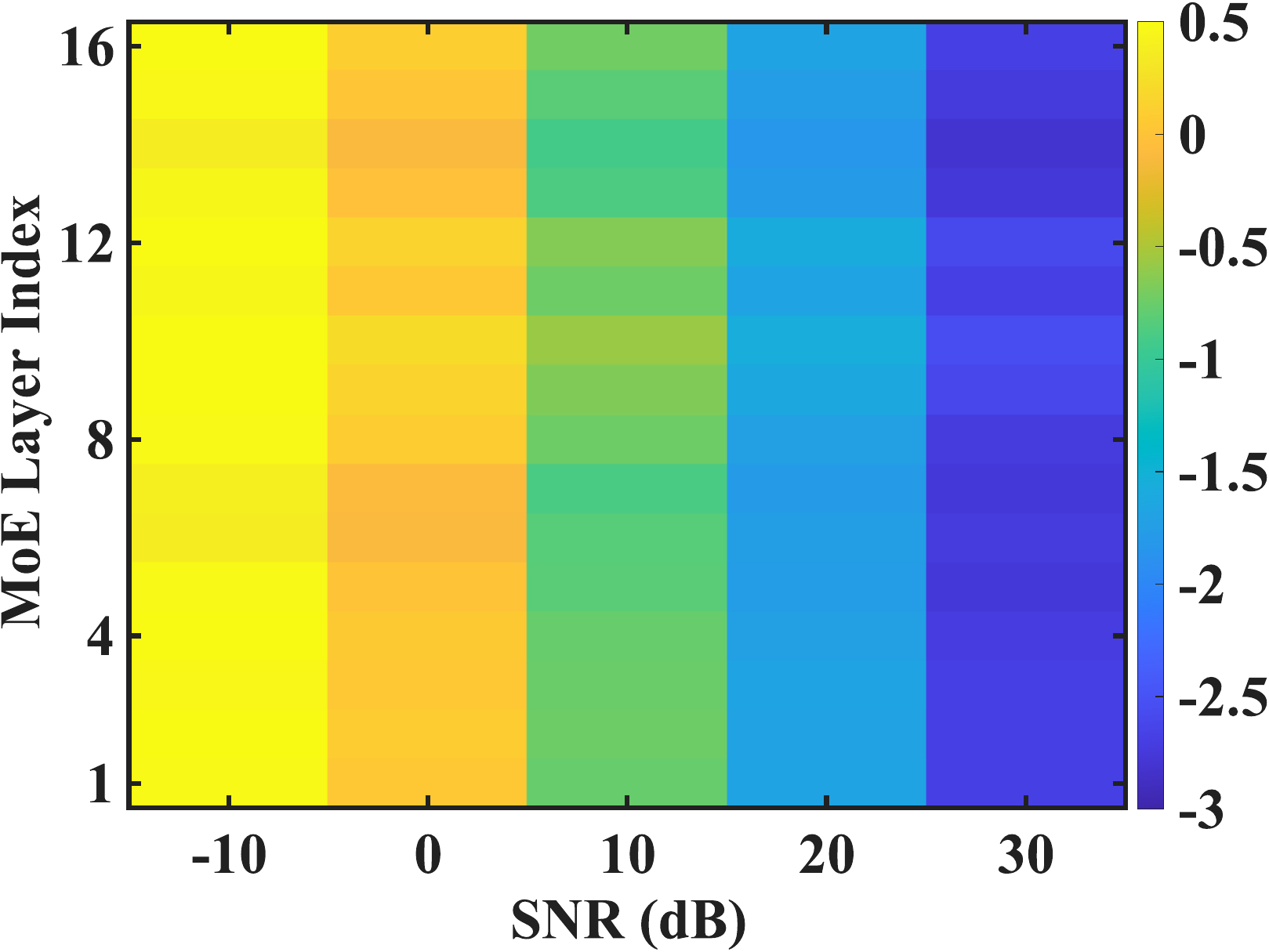}
        \caption{Block--SNR WMSE heatmap.}
        \label{fig:layer_error}
\end{figure}

\begin{figure}[t]
    \centering
 \includegraphics[width=0.75\linewidth,,height = 0.25\textwidth]{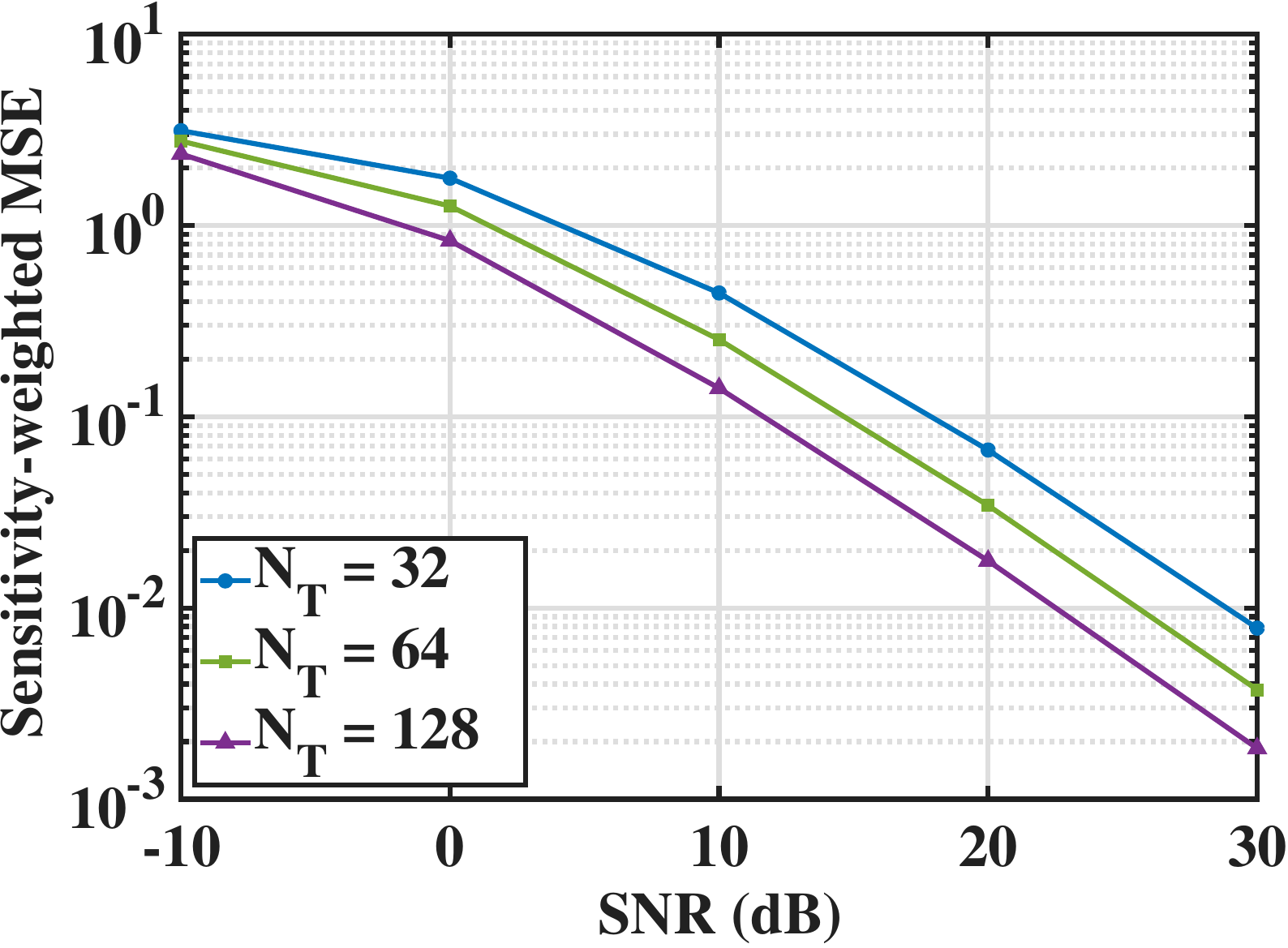}
        \caption{Sensitivity-Weighted MSE vs SNR for Antenna Scaling.}
        \label{fig:mimo}
\end{figure}

\begin{figure}[!t]
    \centering
    \subfigure[WMSE ratio at $N_T=128$]{
        \includegraphics[width=0.35\textwidth,height = 0.23\textwidth]{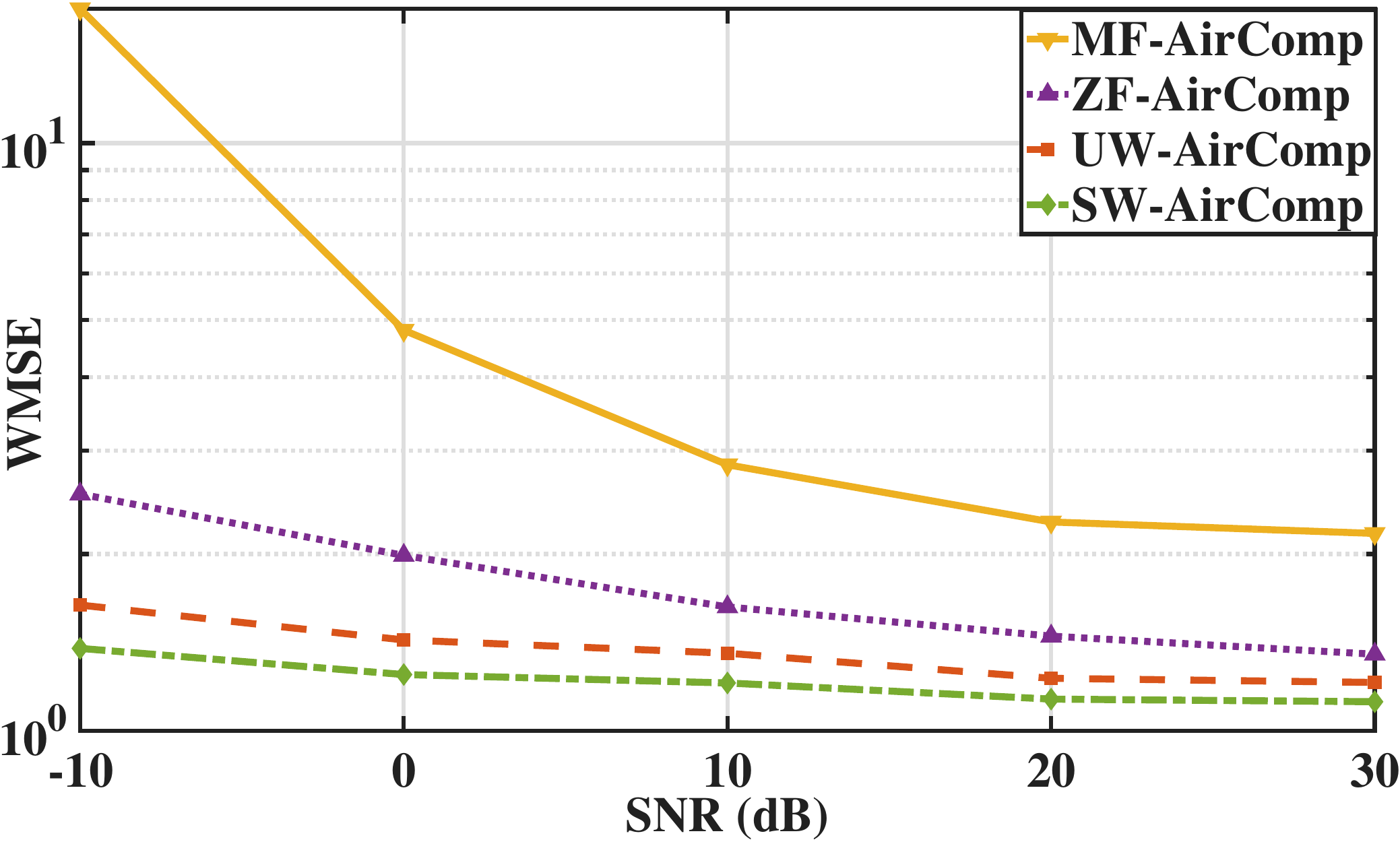}}\label{fig:mimo_2}
        \quad
    \subfigure[WMSE reduction at $N_T=128$]{
        \includegraphics[width=0.35\textwidth,height = 0.23\textwidth]{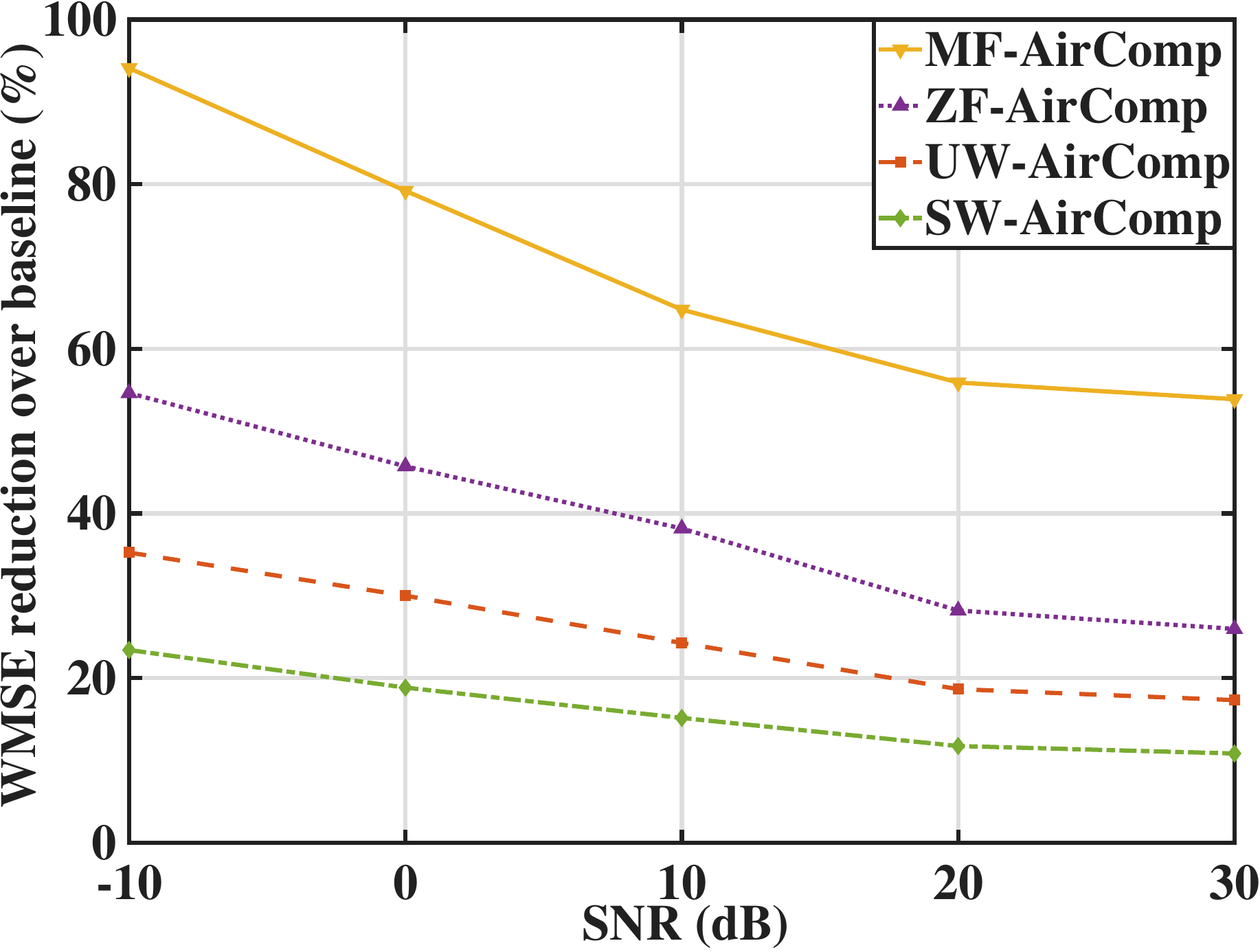}}\label{fig:mimo_3}
    \caption{Massive-MIMO aggregation study over 16 wireless aggregation blocks. }
    \label{fig:massive-mimo-study}
\end{figure}

Fig.~\ref{fig:layer_error} is placed first because it shows the
structural reason for sensitivity-aware allocation: residual distortion is
not uniformly distributed over the sparse-MoE aggregation blocks, so a single
uniform protection rule is too coarse. The 16 integer indices in
Fig.~\ref{fig:layer_error} are wireless aggregation blocks produced by the
packing and sensitivity-reduction procedure. They are optimization blocks in
the wireless simulator, not a statement about the number of MoE layers in
Qwen3. The heatmap also explains why a scalar average-MSE objective is weak:
two designs with similar average distortion can place their residual error on
different blocks, and the block with larger logit sensitivity is more likely to
perturb the output logits.

Fig.~\ref{fig:mimo} verifies the expected antenna-scaling
behavior: increasing $N_T$ improves the weighted aggregation MSE over the
tested SNR range. This is the array-gain effect expected from a larger
transmit aperture: the anchor BS has more spatial degrees of freedom to align the
desired sum of expert outputs and suppress aggregation error under the same power budget.
Fig.~\ref{fig:massive-mimo-study} then compares the transceiver designs at
$N_T=128$. The antenna-scaling result gives a second design lesson. More transmit antennas
provide additional spatial degrees of freedom, but extra antennas do not define
which hidden states deserve protection. A sensitivity-agnostic design can use
array gain to improve average channel matching, yet still spend part of the resource blocks whose errors have little effect on the logits.
Sensitivity weighting changes the role of the array by steering the available
degrees of freedom toward aggregation blocks with larger downstream effect.
This supports coordinating antenna resources, subcarrier assignment, and power
allocation with compact model-side metadata, rather than optimizing only from
CSI and transmit-power budgets. 

Fig.~\ref{fig:massive-mimo-study}(a) reports each baseline's WMSE vs SNR at  $N_T=128$. Fig.~\ref{fig:massive-mimo-study}(b) reports the corresponding WMSE reduction. SW-AirComp uses the same channel realization and power budget as the baselines, but it weights the objective by logit sensitivity. Its lower WMSE therefore means that the gain is not only from improving the physical channel fit; it also comes from moving the residual error away from blocks that matter more to the next-token logits. UW-AirComp isolates the value of sensitivity weighting, MF-AirComp represents channel-matched analog aggregation, and ZF-AirComp represents channel-inversion aggregation that can amplify noise when the effective channel is ill conditioned.

\begin{figure}[t]
    \centering
    \safeincludegraphics[width=0.85\linewidth,height = 0.35\textwidth]{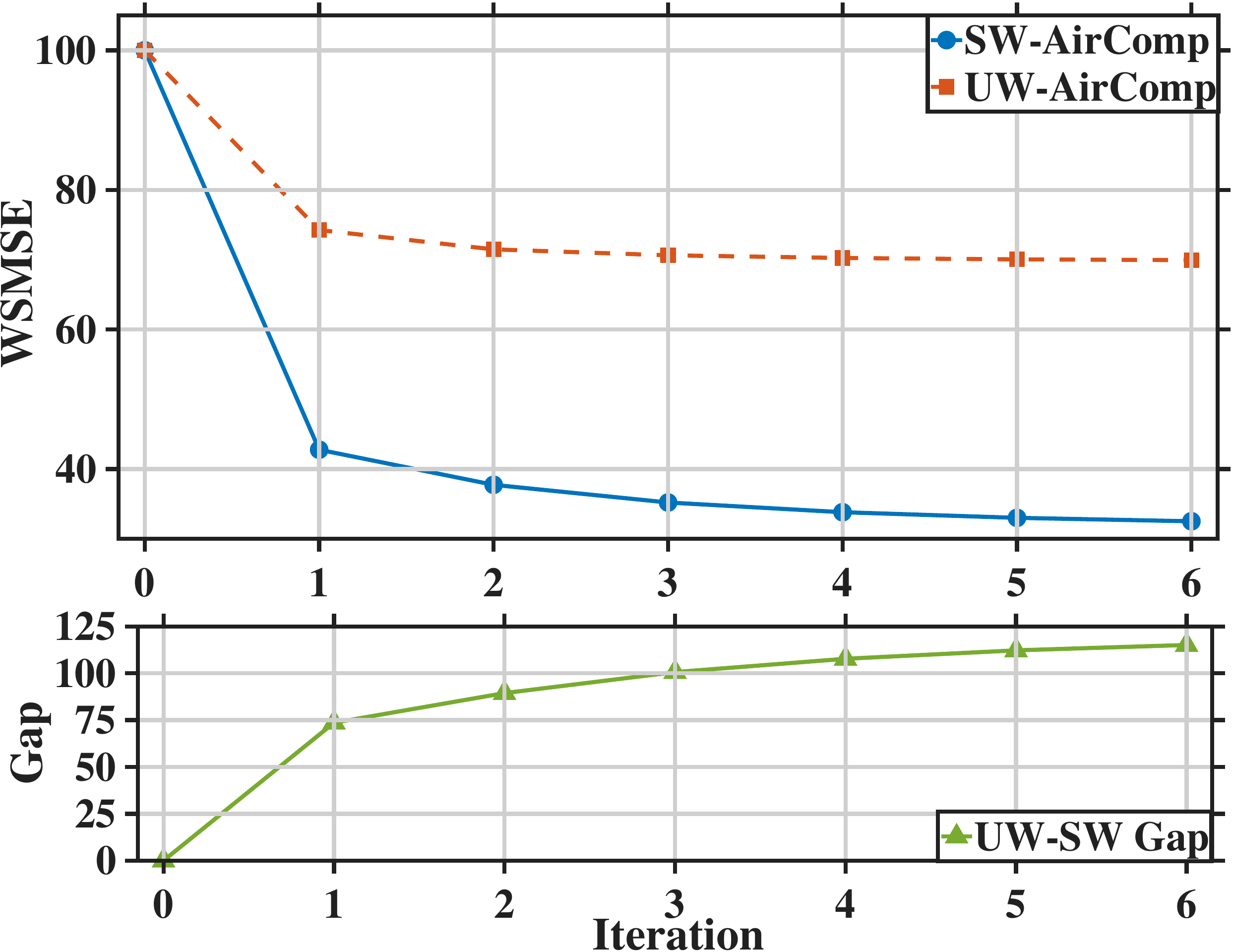}
    \caption{Normalized WSMSE convergence.}
    \label{fig:inner-loop-convergence}
\end{figure}
\begin{figure}[t]
    \centering
    \subfigure[Logit-RMSE ratio]{
        \safeincludegraphics[width=0.35\textwidth,height = 0.26\textwidth]{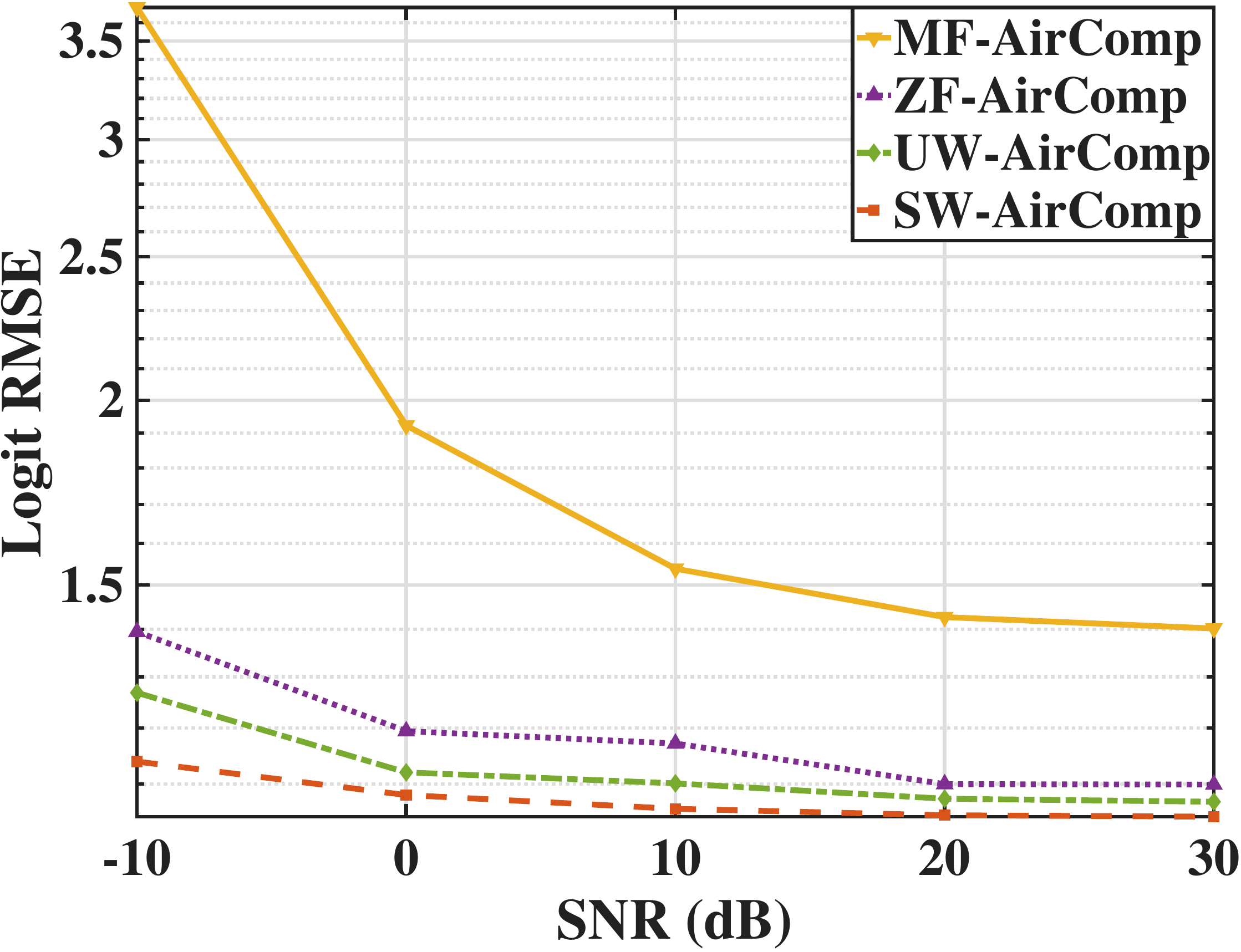}}
    \hfill
    \subfigure[Logit-gap violation rate]{
        \safeincludegraphics[width=0.35\textwidth,height = 0.26\textwidth]{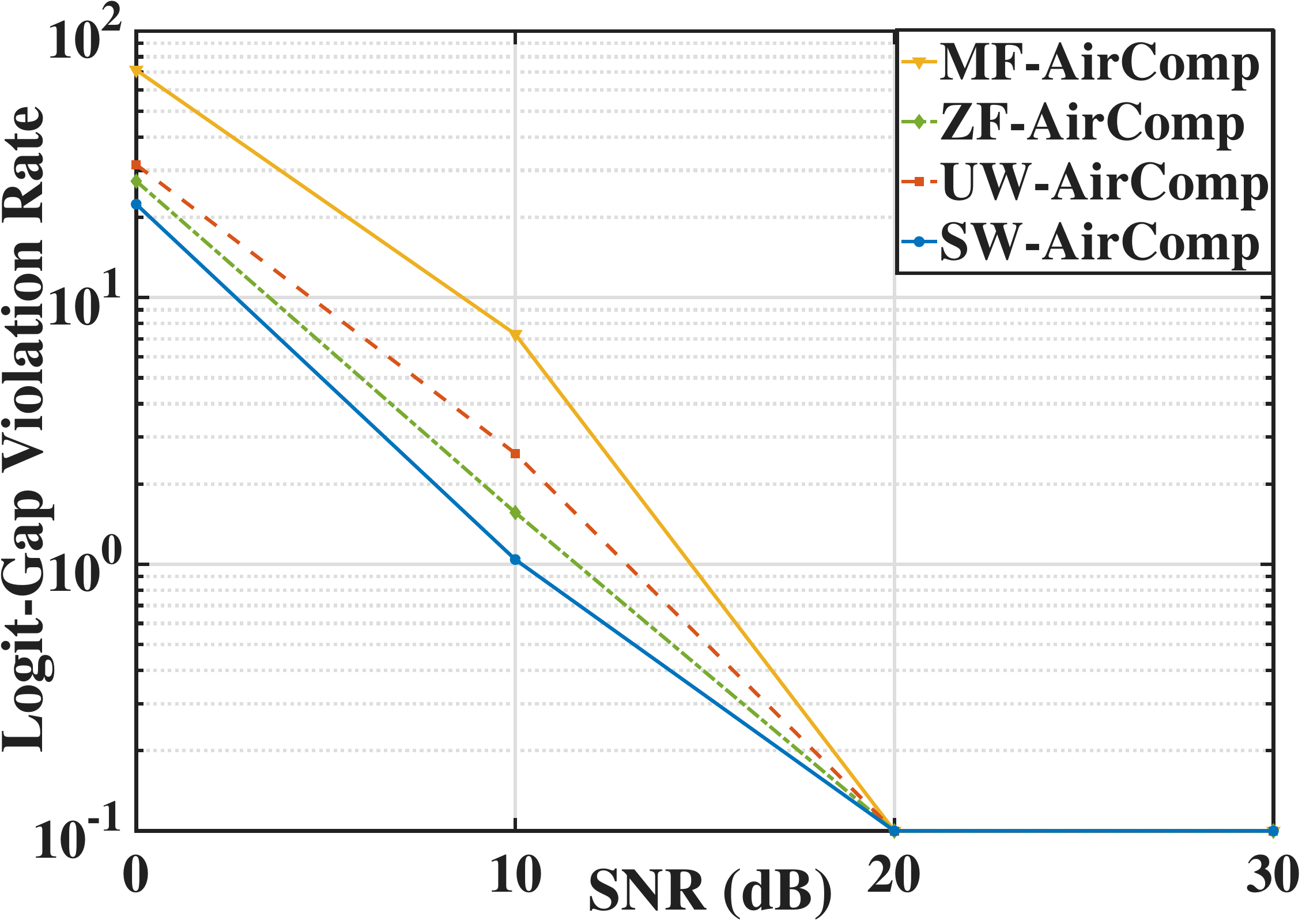}}
    \caption{(a) reports each baseline's weighted-logit RMSE normalized by SW-AirComp. (b) reports the percentage of
    aggregation blocks whose estimated logit perturbation exceeds the clean
    top-token logit gap; lower is better. }
    \label{fig:sensitivity-ablation-diagnostics}
\end{figure}

Fig.~\ref{fig:inner-loop-convergence} reports the fixed-weight convergence
diagnostic at 20 dB and $N_T=8$ with constant channel condition. The plotted metric is WSMSE, $100J_{\rm eval}^{(i)}/J_{\rm eval}^{(0)}$, where $J_{\rm eval}^{(i)}=\sum_b \omega_b\epsilon_b^{(i)}$ for $N_T=8$. Both methods are evaluated with the same sensitivity weights, channels, source covariances,
and samples; hence the plot measures reduction of the logit-aware inference-risk criterion on a single scale. The two methods start at the same point because they use the same matched-filter initialization. After the first transceiver update, SW-AirComp directly descends its WSMSE, whereas UW-AirComp descends its uniform-weight MSE objective and is then evaluated by the same WSMSE. After six transceiver updates, i.e., at iteration 6, the normalized WSMSE decreases from 100 to 32.5 for SW-AirComp and to 69.9 for UW-AirComp. The relative WSMSE gap, $100(J_{\rm UW}-J_{\rm SW})/J_{\rm SW}$, increases from $73.6\%$ after the first update to $115.1\%$ at iteration 6. Most SW-AirComp improvement occurs within the first two iterations. This behavior supports
Proposition~\ref{prop:inner-loop-convergence} for the fixed-objective SW inner solver, while the outer sensitivity refresh remains a tracking loop with a changing objective.

Let $\mathbf h_b$ and $\widehat{\mathbf h}_b^{(a)}$ be the clean and reconstructed real hidden block for method $a$ after inverse packing. The main PHY-layer metric is the sensitivity-weighted aggregation MSE
\begin{equation}
    {\rm WMSE}^{(a)}
    =
    \frac{1}{|\cB|}
    \sum_{b\in\cB}
    \omega_b
    \frac{
    \left\|
    \widehat{\mathbf h}_b^{(a)}
    -
    \mathbf h_b
    \right\|_2^2
    }{d_b},
    \label{eq:wmse-metric}
\end{equation}
where $d_b$ is the hidden dimension represented by block $b$. This metric is the empirical
counterpart of the sensitivity-weighted AirComp objective in
Problem~\eqref{prob:main}. A smaller ${\rm WMSE}^{(a)}$ means that the
reconstructed hidden aggregate is closer to the clean aggregate on the
sensitivity-weighted hidden dimensions, while a larger value means that the
AirComp error is larger or falls on blocks with larger sensitivity weights.

Fig.~\ref{fig:sensitivity-ablation-diagnostics} further uses three
diagnostic metrics. For method $a$ at SNR $\gamma$, the normalized
weighted-logit RMSE is
\begin{equation}
    \rho_{\rm logit}^{(a)}(\gamma)
    =
    \frac{
    {\rm RMSE}_{\rm logit}^{(a)}(\gamma)
    }{
    {\rm RMSE}_{\rm logit}^{({\rm SW})}(\gamma)
    },
    \label{eq:rho-logit}
\end{equation}
where values above one mean that SW-AirComp has lower logit RMSE than method
$a$, i.e., the SW-Aircomp is more advantageous. This metric measures how much the wireless aggregation error changes the
model logits after the hidden perturbation is propagated to the output layer.
The normalization uses SW-AirComp as the reference at the same SNR, so values
above one indicate larger logit distortion than SW-AirComp. The logit-gap violation rate is
\begin{equation}
    p_{\rm gap}^{(a)}(\gamma)
    =
    \frac{1}{|\cB_\gamma|}
    \sum_{b\in\cB_\gamma}
    \mathbf 1
    \!\left\{
    \left\|
    \Delta\bo_b^{(a)}
    \right\|_2
    >
    g_b
    \right\},
    \label{eq:logit-gap-violation-prob}
\end{equation}
where $\cB_\gamma$ is the evaluated block-sample set at SNR $\gamma$,
$\Delta\bo_b^{(a)}$ is the estimated first-order logit perturbation, and
$g_b$ is the clean top-token logit gap. The plotted value of
$p_{\rm gap}^{(a)}$ is reported as a percentage. This metric measures how often
the estimated logit perturbation exceeds the clean logit gap
between the two largest logits. A smaller value means that the wireless aggregation error is less
likely to change the selected token.

Fig.~\ref{fig:sensitivity-ablation-diagnostics} links the wireless WMSE
objective to inference-facing quantities. Fig.~\ref{fig:sensitivity-ablation-diagnostics}(a) reports the normalized
logit-RMSE ratio in \eqref{eq:rho-logit}. Because the denominator is the
SW-AirComp logit RMSE at the same SNR, a curve above one means that the
corresponding baseline perturbs the output logits more than the proposed
method. This is a stronger diagnostic than hidden-state WMSE alone: it asks
whether the remaining LLM computation amplifies the wireless error into the
vocabulary logits that determine the next token.

Fig.~\ref{fig:sensitivity-ablation-diagnostics}(b) reports the logit-gap violation rate in
\eqref{eq:logit-gap-violation-prob}. A higher value means that more block samples estimate logit perturbations exceeding the clean top-token logit gap. The
low-SNR overloaded region is less useful for ranking because most methods
approach saturation; the meaningful regime is the transition to usable SNR,
where SW-AirComp reduces both logit-RMSE and logit-gap violation risk relative to
UW-AirComp, MF-AirComp, and ZF-AirComp. This operating-regime explains why the evaluation reports WMSE together with logit-RMSE and logit-gap violation. These metrics are not interchangeable. A method can reduce unweighted hidden-state MSE while still placing residual error on a block with large logit sensitivity.

\subsection{Insights and Limitations}
\label{subsec:simulation-design-insights}

For edge deployment, the physical layer does not need the full LLM state. It
only needs block-level priority metadata that indicates which aggregation
blocks are more likely to affect the next-token logits. This keeps the wireless
optimizer as a communication module driven by CSI, noise variance, and power
constraints, while the sensitivity weights provide an inference-side priority
signal. In moderate SNR, this metadata improves the allocation of spatial and
power resources.

This perspective differs from AirComp for FL. Training-time
aggregation noise is usually judged through its long-term effect on a learning
trajectory, where later rounds may partly compensate for a noisy update. In
autoregressive MoE inference, an aggregation error affects the immediate hidden
state used to choose the next token, and a wrong token changes the following
context. Logit-gap violation is therefore a decision-level diagnostic, not only a
secondary communication metric.

The main practical risk is the stale or noisy sensitivity metadata. Overestimating
a block can waste power or spatial degrees of freedom, whereas underestimating
it can leave a low-WMSE aggregate that still produces a harmful logit
perturbation. Robust deployments should normalize or clip sensitivity weights,
refresh them at layer-dependent intervals, and retain unweighted AirComp as a
fallback when model-side metadata is unavailable.

The Qwen3/OpenCompass replay
demonstrates that controlled MoE aggregation perturbations can reduce
downstream accuracy and that sufficiently high aggregation SNR recovers most
clean-model performance on the tested subsets. The direct ARC-Challenge
closed-loop audit then inserts the AirComp perturbation paths into Qwen3 token
generation and shows that SW-AirComp gives the highest accuracy among
SW-AirComp, UW-AirComp, MF-AirComp, and ZF-AirComp over the
295 evaluated examples. The wireless simulator shows that the proposed
transceiver reduces sensitivity-weighted distortion under the modeled
resource-block MIMO channels. The logit-RMSE and logit-gap violation
diagnostics bridge these levels by linking hidden-state aggregation error to
output-logit perturbation and token-decision risk. A broader benchmark-level
claim still requires extending the closed-loop audit to more channel seeds,
additional SNRs, additional tasks, and fully online CSI, routing, and
sensitivity calibration.

\section{Conclusion}
\label{sec:conclusion}

This paper developed a logit-aware MIMO AirComp framework for distributed MoE inference over wireless edge networks, where the transmitted symbols represent a linear, invertible complex packing of the hidden-state contributions required by MoE aggregation. By weighting block-level AirComp distortion according to its local impact on next-token logits, the proposed design links PHY-layer aggregation errors to LLM inference quality. The resulting framework combines Jacobian-based sensitivity modeling, block-level error weighting, and power-constrained transceiver optimization to allocate wireless resources according to the semantic importance of each aggregation block. The simulation results show that incorporating logit sensitivity improves robustness over sensitivity-agnostic AirComp baselines under the modeled wireless inference setting; in the direct ARC-Challenge closed-loop audit, SW-AirComp also gives the highest Qwen3 task accuracy among the tested AirComp methods.

\appendices

\section{Proof of Lemma~\ref{lem:scalar-sensitivity-reduction}}
\label{app:proof-scalar-sensitivity}
From \eqref{eq:block-logit-error-sum}, $\widehat{\bo}_t-\bo_t
    \approx
    \sum_{b\in\cB_{t,\ell}}
    \bJ_{\ell,t}\mathbf G_b\br_b$. Taking the expected squared norm yields
\begin{align}
    \E\!\left[
    \left\|
    \widehat{\bo}_t-\bo_t
    \right\|_2^2
    \right]
    &\approx
    \sum_{b\in\cB_{t,\ell}}
    \Tr\!\left(
    \mathbf G_b^T
    \bJ_{\ell,t}^T
    \bJ_{\ell,t}
    \mathbf G_b
    \E[\br_b\br_b^T]
    \right),
\end{align}
because the cross terms vanish by mutual uncorrelatedness of the block
errors. Using \eqref{eq:isotropic-resource-error},
\eqref{eq:block-sensitivity-matrix}, and \eqref{eq:resource-weight},
\begin{equation}
    \E\!\left[
    \left\|
    \widehat{\bo}_t-\bo_t
    \right\|_2^2
    \right]
    \approx
    \sum_{b\in\cB_{t,\ell}}
    \frac{\epsilon_b}{q_b}\Tr(\bOmega_b)
    =
    \sum_{b\in\cB_{t,\ell}}
    \omega_b\epsilon_b .
\end{equation}
This proves the scalar reduction.

\section{Proof of Theorem~\ref{thm:sensitivity-bound}}
\label{app:proof-sensitivity-bound}
The Taylor's theorem with locally $L_{\ell,t}$-Lipschitz Jacobian gives
\begin{equation}
    \widehat{\bo}_t-\bo_t
    =
    \bJ_{\ell,t}\bdelta_t^{(\ell)}
    +
    \be_{\ell,t},
    \qquad
    \|\be_{\ell,t}\|_2
    \le
    \frac{L_{\ell,t}}{2}
    \|\bdelta_t^{(\ell)}\|_2^2 .
\end{equation}
For any $\rho>0$ and conformable vectors $\mathbf a$ and $\mathbf b$,
\begin{equation}
    \|\mathbf a+\mathbf b\|_2^2
    \le
    (1+\rho)\|\mathbf a\|_2^2
    +
    (1+\rho^{-1})\|\mathbf b\|_2^2 .
\end{equation}
Applying this inequality with
$\mathbf a=\bJ_{\ell,t}\bdelta_t^{(\ell)}$ and
$\mathbf b=\be_{\ell,t}$, and then taking expectation, gives
\begin{align}
    \E\!\left[
    \left\|
    \widehat{\bo}_t-\bo_t
    \right\|_2^2
    \right]
    \le&
    (1+\rho)
    \E\!\left[
    \left\|
    \bJ_{\ell,t}\bdelta_t^{(\ell)}
    \right\|_2^2
    \right] \nonumber\\
    &+
    (1+\rho^{-1})
    \E\!\left[
    \left\|
    \be_{\ell,t}
    \right\|_2^2
    \right].
\end{align}
Using \eqref{eq:hidden-error-decomposition}, the first-order term is
\begin{align}
    \E\!\left[
    \left\|
    \bJ_{\ell,t}\bdelta_t^{(\ell)}
    \right\|_2^2
    \right]
    &=
    \E\!\left[
    \left\|
    \sum_{b\in\cB_{t,\ell}}
    \bJ_{\ell,t}\mathbf G_b\br_b
    \right\|_2^2
    \right] \nonumber\\
    &=
    \sum_{b\in\cB_{t,\ell}}
    \Tr\!\left(
    \bOmega_b
    \E\!\left[
    \br_b\br_b^T
    \right]
    \right) =
    \sum_{b\in\cB_{t,\ell}}
    \omega_b\epsilon_b ,
\end{align}
where the cross terms vanish because the block errors are mutually
uncorrelated, and the last equality follows from
\eqref{eq:isotropic-resource-error} and \eqref{eq:resource-weight}. The
Taylor remainder satisfies
\begin{equation}
    \E\!\left[
    \left\|
    \be_{\ell,t}
    \right\|_2^2
    \right]
    \le
    \frac{L_{\ell,t}^2}{4}
    \E\!\left[
    \left\|
    \bdelta_t^{(\ell)}
    \right\|_2^4
    \right].
\end{equation}
Substituting the two bounds above gives
\eqref{eq:sensitivity-weighted-upper-bound}.

\section{Proof of Corollary~\ref{cor:token-preservation}}
\label{app:proof-token-preservation}
If \eqref{eq:token-preservation-condition} holds, then $ \widehat{o}_{t,y_t^\star} 
    >
    o_{t,y_t^\star}
    -
    \frac{\Gamma_t}{2}$ and $\widehat{o}_{t,j}
    <
    o_{t,j}
    +
    \frac{\Gamma_t}{2}
    \le
    o_{t,y_t^\star}
    -
    \frac{\Gamma_t}{2}$, $j\ne y_t^\star$. Thus $\widehat{o}_{t,y_t^\star}>\widehat{o}_{t,j}$ for every
$j\ne y_t^\star$, and the greedy token is unchanged. A token flip can occur
only if $\left\|\widehat{\bo}_t-\bo_t\right\|_\infty\ge\frac{\Gamma_t}{2}$.
Since $\|\cdot\|_\infty\le\|\cdot\|_2$, Markov's inequality gives
\begin{align}
    \Pr\!\left[
    \argmax_j\widehat{o}_{t,j}
    \ne
    y_t^\star
    \right]
    \le
    \Pr\!\left[
    \left\|
    \widehat{\bo}_t-\bo_t
    \right\|_2^2
    \ge
    \frac{\Gamma_t^2}{4}
    \right] \le
    \frac{4D_t}{\Gamma_t^2}.
\end{align}
This proves the result.

\section{Proof of Lemma~\ref{lem:mse-decomposition}}
\label{app:proof-mse-decomposition}
Based on the over-the-air superposition in \eqref{eq:aircomp-received},
\begin{equation}
    \bA_b^H\by_b-\bs_b
    =
    \sum_{m\in\cM_b}
    \left(
    \bA_b^H\bH_{m,b}\bF_{m,b}
    -
    \bI_r
    \right)
    \bs_{m,b}
    +
    \bA_b^H\bw_b .
\end{equation}
Taking squared norm and expectation, the cross terms vanish because the
sources are mutually uncorrelated and independent of the receiver noise.
Therefore,
\begin{align}
    \epsilon_b
    &=
    \sum_{m\in\cM_b}
    \Tr\!\Big(
    \left(
    \bA_b^H\bH_{m,b}\bF_{m,b}
    -
    \bI_r
    \right)
    \bR_{m,b}\nonumber\\
    &\quad
    \left(
    \bA_b^H\bH_{m,b}\bF_{m,b}
    -
    \bI_r
    \right)^H
    \Big)
+
    \sigma_b^2
    \Tr\!\left(
    \bA_b^H\bA_b
    \right),
\end{align}
which is equivalent to \eqref{eq:mse-decomposition}. For fixed transmit precoders, expand
\begin{equation}
    \epsilon_b
    =
    \E\!\left[
    \left\|
    \bA_b^H\by_b-\bs_b
    \right\|_2^2
    \right].
\end{equation}
This gives
\begin{align}
    \epsilon_b
    &=
    \Tr\!\left(
    \bA_b^H
    \E[\by_b\by_b^H]
    \bA_b
    \right)
    -
    2\Re\!\left\{
    \Tr\!\left(
    \bA_b^H
    \E[\by_b\bs_b^H]
    \right)
    \right\}
    +
    \E[\|\bs_b\|_2^2].
\end{align}
By the source model,
$\E[\by_b\by_b^H]=\bR_{y,b}$ and
$\E[\by_b\bs_b^H]=\bR_{ys,b}$, where
$\bR_{y,b}$ and $\bR_{ys,b}$ are defined in
\eqref{eq:Ry} and \eqref{eq:Rys}. The last term is independent of
$\bA_b$, so the fixed-precoder receive block is \eqref{prob:rx-block}.
Because $\bR_{y,b}\succ\bzero$ and $\omega_b>0$, this block is strongly
convex. Differentiating with respect to $\bA_b^\ast$ gives
\begin{equation}
    \bR_{y,b}\bA_b-\bR_{ys,b}=\bzero,
\end{equation}
whose unique solution is \eqref{eq:optimal-combiner}.

\section{Proof of Proposition~\ref{prop:precoder}}
\label{app:proof-precoder}
For fixed receive combiners, the map
\[
    \bF_{m,b}
    \mapsto
    \left(
    \bA_b^H\bH_{m,b}\bF_{m,b}
    -
    \bI_r
    \right)
    \bR_{m,b}^{1/2}
\]
is affine in $\bF_{m,b}$. Hence the squared Frobenius norm in
\eqref{prob:tx-block-obj} is convex, and the nonnegative weighted sum over
$b\in\cB_m$ remains convex. The constraint
\eqref{prob:tx-block-total} is a convex quadratic inequality because
$\bR_{m,b}\succeq\bzero$. Therefore, \eqref{prob:tx-block} is a convex
QCQP. The Lagrangian of the BS-$m$ problem is
\begin{align}
    \mathcal L_m
    =
    &\sum_{b\in\cB_m}
    \omega_b
    \left\|
    \left(
    \bA_b^H\bH_{m,b}\bF_{m,b}
    -
    \bI_r
    \right)
    \bR_{m,b}^{1/2}
    \right\|_F^2
    \nonumber\\
    &+
    \lambda_m
    \left(
    \sum_{b\in\cB_m}
    \Tr\!\left(
    \bF_{m,b}\bR_{m,b}\bF_{m,b}^H
    \right)
    -
    P_m^{\rm tot}
    \right).
\end{align}
The stationarity condition with respect to $\bF_{m,b}^{\ast}$ is
\begin{align}
    \omega_b
    \bH_{m,b}^H\bA_b
    \left(
    \bA_b^H\bH_{m,b}\bF_{m,b}
    -
    \bI_r
    \right)
    \bR_{m,b}
    +
    \lambda_m
    \bF_{m,b}\bR_{m,b}
    =
    \bzero .
    \label{eq:precoder-stationarity-app}
\end{align}
On the transmitted source subspace of $\bR_{m,b}$, rearranging
\eqref{eq:precoder-stationarity-app} gives
\begin{equation}
    \left(
    \bH_{m,b}^H\bA_b\bA_b^H\bH_{m,b}
    +
    \frac{\lambda_m}{\omega_b}
    \bI_{N_{\rm T}}
    \right)
    \bF_{m,b}
    =
    \bH_{m,b}^H\bA_b .
\end{equation}
Thus the optimal precoder has the form
\eqref{eq:optimal-precoder-weighted}. The remaining KKT conditions are
\[
    \lambda_m\ge0,\qquad
    P_m(\lambda_m)\le P_m^{\rm tot},\qquad
    \lambda_m\big(P_m(\lambda_m)-P_m^{\rm tot}\big)=0 .
\]
Therefore, if $P_m(0)\le P_m^{\rm tot}$, the precoder computed with
$\lambda_m=0$ already satisfies the BS power budget, and thus
$\lambda_m^\star=0$. If $P_m(0)>P_m^{\rm tot}$, the precoder would exceed
the BS power budget. In this case, $\lambda_m^\star$ is selected so that
$P_m(\lambda_m^\star)=P_m^{\rm tot}$. This gives
\eqref{eq:precoder-dual-selection}. Since \eqref{prob:tx-block} is convex,
these KKT conditions are sufficient for global optimality. If the matrix in \eqref{eq:optimal-precoder-weighted} is nonsingular, the
Moore--Penrose inverse is replaced by the ordinary inverse.

\section{Proof of Proposition~\ref{prop:inner-loop-convergence}}
\label{app:proof-inner-loop-convergence}
By Lemma~\ref{lem:mse-decomposition}, $J$ is the weighted sum of the block MSE in
\eqref{prob:main-obj}. For fixed $\bF^{(i)}$, the receive-combiner subproblem
decouples over $b$ and is minimized by \eqref{eq:optimal-combiner}; hence
\begin{equation}
    J(\bF^{(i)},\bA^{(i+1)})
    \le
    J(\bF^{(i)},\bA^{(i)}).
    \label{eq:proof-rx-descent}
\end{equation}
For fixed $\bA^{(i+1)}$, Proposition~\ref{prop:precoder} shows that
\eqref{eq:optimal-precoder-weighted} with $\lambda_m$ selected by
\eqref{eq:precoder-dual-selection} globally solves each separated BS transmit
subproblem, so
\begin{equation}
    J(\bF^{(i+1)},\bA^{(i+1)})
    \le
    J(\bF^{(i)},\bA^{(i+1)}).
    \label{eq:proof-tx-descent}
\end{equation}
Combining \eqref{eq:proof-rx-descent} and \eqref{eq:proof-tx-descent} gives
\eqref{eq:inner-loop-descent}. Nonnegativity and $\omega_b>0$ lower-bound the
objective by zero, so the monotone objective values converge. If an
accumulation point admitted a strict single-block decrease, the continuity of $J$
and the exact block minimizations would preserve that decrease nearby,
contradicting convergence along the subsequence. Thus every accumulation point
is block-stationary.

\bibliographystyle{IEEEtran}
\bibliography{references}

\end{document}